\documentclass[twocolumn,journal]{IEEEtran}

\usepackage{amsmath,amssymb,array}
\usepackage{algorithm}
\usepackage{algorithmic}
\usepackage{multirow}
\usepackage{color}
\usepackage{bm}
\usepackage[pdftex]{graphicx}

\usepackage{url}
\DeclareMathOperator{\Tr}{\mathrm{Tr}}

\newtheorem{theorem}{Theorem}

\newtheorem{lemma}{Lemma}%
\newtheorem{corollary}{Corollary}%

\newcommand{\argmin}{\mathop{\rm argmin}\limits}

\def\supp{\mathop{\rm supp}}

\def\X{{\mathcal{X}}}
\def\Y{{\mathcal{Y}}}

\def\H{{\mathcal{H}}}

\def\P{{\mathcal{P}}}

\def\QED{\mbox{\rule[0pt]{1.5ex}{1.5ex}}}

\newcommand{\qed}{\hfill \QED}

\def\Label#1{\label{#1}\ [\ \text{#1}\ ]\ }
\def\Label{\label}

\begin{document}
\title{Analytical algorithm for capacities of classical and classical-quantum channels}
\author{
Masahito~Hayashi~\IEEEmembership{Fellow,~IEEE}
\thanks{
The author was supported in part by the National Natural Science Foundation of China (Grant No. 62171212) and
Guangdong Provincial Key Laboratory (Grant No. 2019B121203002).
The material in this paper was
presented in part at the 2022 IEEE International Symposium on Information
Theory (ISIT 2022), Aalto University, Espoo, Finland, 26 June -- 1 July, 2022. \cite{ISIT2022}.}
\thanks{Masahito Hayashi is with 
Shenzhen Institute for Quantum Science and Engineering, Southern University of Science and Technology, Nanshan District,
Shenzhen, 518055, China,
International Quantum Academy (SIQA), Futian District, Shenzhen 518048, China,
Guangdong Provincial Key Laboratory of Quantum Science and Engineering,
Southern University of Science and Technology, Nanshan District, Shenzhen 518055, China,
and
Graduate School of Mathematics, Nagoya University, Nagoya, 464-8602, Japan.
(e-mail:hayashi@sustech.edu.cn, masahito@math.nagoya-u.ac.jp)}}
\date{}
\markboth{M. Hayashi: Analytical algorithm for capacities of classical and classical-quantum channels}{}

\maketitle
\begin{abstract}
We derive an analytical algorithm for the channel capacity of a classical channel without any iteration, 
while its existing algorithms require iterations and the number of iterations depends on the required precision level.
Hence, our algorithm is its first analytical algorithm for this task without any iteration,
while this algorithm needs several conditions for the channel. 
We apply the obtained algorithm to examples, and see 
how the obtained algorithm works in these examples.
Then, we extend it to the channel capacity of a classical-quantum (cq-) channel.
Many existing studies proposed algorithms for a cq-channel and
all of them require iterations.
Our extended analytical algorithm has also no iteration, and outputs the exactly optimum value.
\end{abstract}


\section{Introduction}\Label{S1}
One of the key problems in classical and quantum information theory is the maximization of information quantities.
However, it is not so easy to perform such a maximization analytically because 
all of existing methods require certain iterations, whose number depends on the required precision level.
The most common maximization problem is 
the channel capacity, which is given as the maximization of mutual information \cite{Shannon},
and its calculation has been studied by Arimoto \cite{Arimoto}, Blahut \cite{Blahut},
and their related studies \cite{Matz,Yu,SSML}.
However, these are iterative approximation algorithms to calculate the maximum of the mutual information.
In addition, the reference \cite{Chiang} calculated only its upper bound and
the references \cite{Huang,NWS} developed other type of method to approximately calculate it.
Hence, they cannot calculate the exact value for the channel capacity.
As variants, the references \cite{Yasui,YSM} extended the above method to 
the wire-tap capacity \cite{Wyner,CK79} when the wire-tap channel is degraded.
Also, the references \cite{Nagaoka,Dupuis,Sutter,Li,RISB} extended it to 
the quantum setting, so called the capacity of classical-quantum channel.
However, these results are also iterative approximation algorithms.

%In this paper, inspired by Toyota \cite{Shoji}, we focus on an information-geometrical structure \cite{Amari-Nagaoka}.Then, 
This paper proposes an algorithm to analytically calculate the channel capacity of the classical channel
without iteration.
The proposed algorithm is composed of 
solving simultaneous linear equations and calculation of logarithm and exponential
because it employs an information-geometrical structure.
However, the proposed method works under certain conditions.
Since our method is analytical, 
we can derive several analytical formulas for the capacity when these conditions are satisfied.
Then, to see this possibility,
we apply our algorithm to several examples, and derive analytical expressions
of the capacities in these examples. 
Further,
we extend our analytical algorithm to the calculations of the capacity of classical-quantum channel.

The remaining part of this paper is organized as follows.
First, Section \ref{S10} derives our algorithm for the capacity of a classical channel.
Section \ref{S3} applies the obtained result to several examples.
Next, Section \ref{S12} extends this method to the capacity of classical-quantum channel.
Finally, 
Section \ref{SC} discusses the merit and the demerit of our method over existing methods.
%because our algorithms need very restrictive conditions and requires a large amount of calculation.

\section{Capacity of classical channel }\Label{S10}
We consider the input and output alphabets
$\X:=\{1,\ldots,n_1 \}$ and $\Y:= \{1,\ldots, n_2 \}$ that are finite sets. 
We denote the sets of probability distributions on $\X$ and $\Y$ by
${\cal P}_{\X}$ and ${\cal P}_{\Y}$, respectively.
For distributions $P,Q \in {\cal P}_{\X}$, 
the entropy $H(P)$ and the divergence $D(P\|Q)$ are defined as
\begin{align}
H(P):= &-\sum_{x \in \X} P(x)\log P(x),  \\
D(P\|Q):= &\sum_{x \in \X} P(x)\log \frac{P(x)}{Q(x)}.
\end{align}
Throughout this paper, the base of the logarithm is chosen to be 
the natural logarithm.

A channel from $\X$ to $\Y$ is given as conditional distribution on $\Y$ conditioned with $\X$.
That is, using the notation $W_x(y):=W(y|x)$, it can be considered as 
a map  $ W  :\X  \rightarrow  {\cal P}_{\Y}$.
For $Q_X \in {\cal P}_{\X}$ and $Q_Y \in {\cal P}_{\Y}$, 
$W \cdot Q_X \in {\cal P}_{\Y}$,
$W \times Q_X \in {\cal P}_{\X \times \Y}$, and $Q_X\times Q_Y \in \mathcal{P}_{\X \times \Y}$ 
are defined by 
$(W \cdot Q_X)(x,y):=\sum_{x \in \X}W(y|x)Q_X(x)$, 
$(W \times Q_X)(x,y):=W(y|x)Q_X(x)$, and $(Q_Y \times Q_X)(x,y) := Q_X(x)Q_Y(y)$, respectively. 

The channel capacity of a channel  $W$ is given by \cite{Shannon}, \cite[p.124]{CK}
\begin{align}
C(W):=&\max_{Q_X \in {{\cal P}_{\X}} }
\sum_{x\in \X}Q_X(x) D(W_x \| W \cdot Q_X  ) \nonumber \\
=& \min_{Q_Y \in {\cal P}_{\Y} }\max_{x \in \X} D( W_x\| Q_Y) \nonumber \\
=& \min_{Q_X \in {\cal P}_{\X} }\max_{x \in \X} D( W_x\| W\cdot Q_X) 
%=\max_{q \in {{\cal P}_{\X}} } \min_{ q'\in {{\cal P}_{\X}} , q'' \in {{\cal P}_{\Y}} }
% D(W \times q \| q'' \times q' )
.\Label{MOA}
\end{align}
To discuss $C(W)$, we assume the following conditions.
\begin{description}
\item[(A)]
$W_1, \ldots, W_{n_1}$ are linearly independent.
\end{description}
Then, we have the following:
\begin{lemma}\Label{LAP}
%Assume Condition (A). %that $W_1, \ldots, W_{n_1}$ are linearly independent.
When a distribution $Q_Y=W\cdot Q_X$ realizes the minimum in \eqref{MOA}, 
it satisfies the following condition:
$ D( W_x\| Q_Y) $ does not depend on $x \in \supp(Q_X)$.
\end{lemma}
Lemma \ref{LAP} is shown in Appendix \ref{A1}.
We define the set ${\cal M}_{0}$ as
\begin{align}
{\cal M}_{0}:= \Big\{ Q_Y \in \P_{\Y}\Big|
Q_Y= \sum_{x \in \X } c(x)W_{x}, \quad
\sum_{x \in \X } c(x)=1\Big\}.
\end{align}
Here, the condition $c(x)\ge 0$ is not imposed.
Hence, ${\cal M}_{0}$ is characterized by as linear constraints, which will be 
explained in Appendix \ref{ALL2}.
Then, we introduce another condition for the distribution $Q_Y$:
\begin{description}
\item[(B)]
$ D( W_x\| Q_Y) $ does not depend on $x \in \X$.
\end{description}

\begin{lemma}\Label{LL2}
When Condition (A) holds, 
only one distribution $Q_Y \in {\cal M}_{0}$ satisfies Condition (B).
\end{lemma}
Lemma \ref{LL2} is shown in Appendix \ref{ALL2}.

In the following, we denote the element of ${\cal M}_{0}$ to satisfy the condition (B) by  $Q_{Y,*}$.
Since $Q_{Y,*}$ belongs to ${\cal M}_{0}$,
there exists a function $\widehat{Q}_{X,*}$ on $\X$ as the solution of the following equation:
\begin{align}
\sum_{x\in \X} W(y|x)\widehat{Q}_{X,*}(x)= Q_{Y,*}(y).\Label{XLR}
\end{align}
Condition (A) guarantees the uniqueness of $\widehat{Q}_{X,*}$.
We have the following theorem:
\begin{theorem}\Label{T0}
Assume Condition (A).
The following conditions are equivalent
\begin{description}
\item[(i)]The relation $D( W_x\| Q_{Y,*})=C(W)$ holds.
\item[(ii)] The function $\widehat{Q}_{X,*} $ satisfies
the condition
\begin{align}
\widehat{Q}_{X,*}(x) \ge 0 \hbox{ for }x \in \X.
\Label{MMR}
\end{align}
%There exists a distribution $Q_X$ on $\X$ such that $ Q_{Y,*}=W\cdot Q_X $.
\end{description}
%In addition, when an element $x \in \X$ does not satisfy \eqref{MMR}, 
%the relation $D( W_x\| Q_{Y,*}) >C(W)$ holds.
\end{theorem}
Theorem \ref{T0} is shown in Appendix \ref{A2}.

When the function $\widehat{Q}_{X,*}$ does not satisfy \eqref{MMR},
it is not a distribution on $\X$.
Due to Theorem \ref{T0}, the condition (i) does not hold.
That is, there exists an element $x \in \X$ such that 
$D( W_x\| Q_{Y,*}) >C(W)$.

Hence, under the condition (ii), the capacity $C(W)$ is given by $D( W_x\| Q_{Y,*})$.
To consider the case that the condition (ii) does not hold,
we prepare the following theorem.
For any function $f$ on $\X$,
we define ${\cal N}(f):= \{x \in \X| f(x) < 0\}$ and 
${\cal N}^c(f):= \{x \in \X| f(x) \ge 0\}$.
\begin{theorem}\Label{T-1}
Assume Condition (A). Then, we have
\begin{align}
C(W)=\max_{Q_X \in {{\cal P}_{{\cal N}^c(\widehat{Q}_{X,*})}} }
\sum_{x\in {\cal N}^c( \widehat{Q}_{X,*})}Q_X(x) D(W_x \| W \cdot Q_X  ) .\Label{ZP1}
\end{align}
\end{theorem}
Theorem \ref{T-1} is shown in Appendix \ref{A3}.
Therefore, the capacity $C(W)$ is obtained only with the input set ${\cal N}^c(\widehat{Q}_{X,*})$.
That is, the function $\widehat{Q}_{X,*}$ gives an important information for computing $C(W)$.

%Hence, when a distribution $Q_Y$ satisfies (B), 
%$ D( W_x\| Q_Y) $ realizes the capacity $C(W)$ \cite[Lemma 1]{NK}.

%Here, for $p \in {\cal P}_{\X \times \Y}$, $(p)_{\Y}$ denotes the marginal distribution on $\Y$. 

We choose $n_2-1$ linearly independent functions $f_1, \ldots, f_{n_2-1}$ on $\Y$ such that they are not constant function and
\begin{align}
\sum_{y \in \Y}W_{n_2} (y)f_j(y)=0 \Label{CO}
\end{align}
for $j=1, \ldots, n_2-1$.
We define the matrix $(h_{i,j})$
\begin{align}
h_{i,j}:= \sum_{y \in \Y}W_i(y) f_j(y).\Label{XPA}
\end{align}
\if0
Since functions $f_1, \ldots, f_{n_2-1}$ are linearly independent
and $W_1, \ldots, W_{n_2}$ are linearly independent,
the condition \eqref{CO} guarantees that 
the matrix $(h_{i,j})_{1\le i,j\le n_2-1}$ is invertible.
\fi
Given an $n_2-1$-dimensional parameter $\theta=(\theta^1, \ldots, \theta^{n_2-1})$,
we define the distribution $P_{\theta,Y}$ as
\begin{align}
P_{\theta,Y}(y)=e^{\sum_{j=1}^{n_2-1} f_j(y)\theta^j-\phi(\theta)  },\Label{nat}
\end{align}
where 
\begin{align}
\phi(\theta):= \log 
\bigg(\sum_{y\in\Y}e^{\sum_{j=1}^{n_2-1} f_j(y)\theta^j} \bigg) \Label{XZP}.
\end{align}
The parameterization \eqref{nat} is called the natural parameter \cite{Amari-Nagaoka}.

We have the following theorem.
\begin{theorem}\Label{T1}
Assume that the parameters $\theta^1, \ldots, \theta^{n_2-1}$ satisfy
the condition
\begin{align}
\sum_{j=1}^{n_2-1}h_{i,j} \theta^j= -H(W_i)+H(W_{n_1}).\Label{MXP}
\end{align}
for $i=1, \ldots, n_1-1$.
Then, we have
\begin{align}
D(W_x\| P_{\theta,Y})=\phi(\theta)-H(W_{n_1}).
\end{align}
for $x\in \X$.
%Since the matrix $(h_{i,j})_{1\le i,j\le n_2-1}$ is invertible, only one set of parameters 
%$\theta^1, \theta^2, \ldots, \theta^{n_2-1}$. 
\end{theorem}

\begin{proof}
The condition \eqref{MXP} implies that
\begin{align}
\sum_{y \in \Y}W_i(y)  \sum_{j=1}^{n_2-1} f_j(y)\theta^j
=&\sum_{j=1}^{n_2-1}h_{i,j} \theta^j\nonumber \\
=&-H(W_i)+H(W_{n_2}).
\end{align}

For $x (\neq n_2) \in \X$, we have
\begin{align}
& D(W_x\| P_{\theta,Y})=
\sum_{y \in \Y}W_x(y) \big( \log W_x(y)-\log P_{\theta,Y}(y) \big) \nonumber \\
=&-H(W_x)
-\sum_{y \in \Y}W_x(y) \bigg(
\sum_{j=1}^{n_2-1} f_j(y)\theta^j-\phi(\theta) \bigg) \nonumber \\
=&-H(W_x)- \Big(-H(W_x) +H(W_{n_2}) -\phi(\theta) \Big)\nonumber \\
=&\phi(\theta)-H(W_{n_2}).
\end{align}
Also, we have
\begin{align}
& D(W_n\| P_{\theta,Y})=
\sum_{y \in \Y}W_{n_2}(y) (\log W_{n_2}(y)-\log P_{\theta,Y}(y)) \nonumber \\
=&-H(W_{n_2})
-\sum_{y \in \Y}W_x(y) \bigg(
\sum_{j=1}^{n_2-1} f_j(y)\theta^j-\phi(\theta) \bigg) \nonumber \\
=&-H(W_{n_2})- \big(-\phi(\theta) \big)
=\phi(\theta)-H(W_{n_2}).
\end{align}
\end{proof}

We define the set ${\cal E}_0$ as
\begin{align}
{\cal E}_0:= \{P_{\theta,Y} | \hbox{ The condition \eqref{MXP} holds.} \}
\end{align}
\begin{lemma}\Label{LX1}
The set ${\cal M}_0 \cap {\cal E}_0$ is composed of one element $P_{\theta_*,Y}$.
\end{lemma}
Lemma \ref{LX1} is shown in Appendix \ref{A4}.
Therefore, $P_{\theta_*,Y}$ equals $Q_{Y,*}$.

Now, as a stronger assumption than Condition (A), 
we assume the following condition (Condition (C)).
\begin{description}
\item[(C)]
$n_1=n_2$ and $W_1, \ldots, W_{n_1}$ are linearly independent.
\end{description}
Since ${\cal M}_{0}=\P_{\X}$, 
due to Lemma \ref{LL2},
only one set of parameters $\theta^1, \ldots, \theta^{n_2-1}$ satisfies
the condition \eqref{MXP}.
Due to Theorem \ref{T1}, solving the equation \eqref{MXP}, we find $Q_{Y,*}$ as $P_{\theta,Y}$.
%In the following, using Theorem \ref{T1}, we consider how to find $Q_{Y,*}$ under Condition (C).
To construct our algorithm, we add the $n_2$-th function $f_{n_2}$ on $\Y$
and define $h_{i,j}$ by \eqref{XPA} for $i,j=1, \ldots, n_2$.
We rewrite the equation \eqref{XLR} as
\begin{align}
\sum_{x\in \X} \widehat{Q}_{X,*}(x) h_{x,j}= &
\sum_{x \in \X} \widehat{Q}_{X,*}(x) \sum_{y \in Y} W_x(y) f_{j}(y) \nonumber \\
=&\sum_{y \in \Y} P_{\theta ,Y}(y)f_{j}(y).
\Label{MMR6}
\end{align}
We obtain the function $\widehat{Q}_{X,*} $ on $\X$ as the solution of \eqref{MMR6},
and $W\cdot \widehat{Q}_{X,*}=Q_{Y,*}$ satisfies the condition (B).
%the maximum \eqref{MOA}, 
When the function $\widehat{Q}_{X,*} $ satisfies
the condition \eqref{MMR},
the value $D(W_x\| P_{\theta ,Y}) $ is the capacity of the channel $W$ due to Theorem \ref{T1}.
Therefore, we have Algorithm \ref{protocol1} to compute $C(W)$ under Condition (C).

In fact, $(W_{i}(j))_{i,j}$ and $(f_{j}(i))_{i,j}$ form $n_2 \times n_2 $ matrices.
When $(f_{j}(i))_{i,j}$ is the inverse matrix of $(W_{i}(j))_{i,j}$, 
$(h_{i,j})_{i,j}$ is the identity matrix.
Due to Theorem \ref{T0},
Theorem \ref{T1} does not necessarily work for calculating $C(W)$.
Hence, based on Theorems \ref{T0} and \ref{T1}, 
we propose Algorithm \ref{protocol1} to check the condition in Theorem \ref{T0},
and compute $C(W)$ under this condition.

%we propose our algorithm to 
In Algorithm \ref{protocol1}, Step 1 has calculation complexity $O(n_2^3)$.
Steps 2 and 3 have calculation complexity $O(n_2^2)$ because $h_{i,j}$ is an upper triangle matrix.
Step 5 has calculation complexity $O(n_2^2)$.
Hence, the total calculation complexity is $O(n_2^3)$.

\begin{algorithm}
\caption{Exact algorithm for classical channel capacity}
\Label{protocol1}
\begin{algorithmic}
\STATE {Step 1: Choose $f_1, \ldots, f_{n_2}$ such that 
$(f_{j}(i))_{i,j}$ is the inverse matrix of $(W_{i}(j))_{i,j}$.
Hence, $h_{i,j}=\delta_{i,j}$.
%$h_{i,j}$ is an upper triangle matrix 
%by using Gaussian elimination. The condition for an upper triangle matrix implies \eqref{CO}.
} 
\STATE {Step 2: Set the parameter $\theta^i= -H(W_i)+H(W_{n_2})$ for $i=1, \ldots, n_2-1$, 
which is the solution of \eqref{MXP}.}
\STATE {Step 3: Calculate $\phi(\theta)$ by using \eqref{XZP}.}
\STATE {Step 4: Calculate $\widehat{Q}_{X,*}(x):= \sum_{y\in \Y} P_{\theta ,Y}(y)f_x(y)$,
where $P_{\theta ,Y}(y)$ is calculated by \eqref{nat}.
This step follows from \eqref{MMR6}.}
\STATE {Step 5: If the condition \eqref{MMR} holds,
%$\widehat{Q}_{X,*}(x)\ge 0 $ for $x \in \X$, 
we consider that the condition in Theorem \ref{T0} holds and
output $\phi(\theta)-H(W_n)$ as the capacity.
Otherwise, we consider that the condition in Theorem \ref{T0} does not hold
and output ``the capacity cannot be computed.''}
\end{algorithmic}
\end{algorithm}

Next, instead of Condition (C), we consider the following condition. 
\begin{description}
\item[(C')]
The relation $n_1 \ge n_2$ holds. 
Any $n_2$ elements among $W_1, \ldots, W_{n_1}$ are linearly independent.
\end{description}
Under this condition, 
we can apply Algorithm \ref{protocol1} to any
$n_2$ elements $x_1, \ldots, x_{n_2}$ in $\X$.
If the capacity is calculated under this choice, 
it is denoted by $C(W;x_1, \ldots, x_{n_2})$.
When the capacity is calculated under all choices of $x_1, \ldots, x_{n_2}$, 
the maximum of $C(W;x_1, \ldots, x_{n_2})$ is the capacity of the channel.

In this case, we need to try ${n_1 \choose n_2}$ combinations, which requires too large 
calculation amount.
However, it is possible to avoid such repetition as follows.
First, we apply the conventional iterative algorithm by \cite{Arimoto,Blahut} or 
the improved iterative algorithm by \cite{SSML}.
Then, we obtain an approximately optimal input distribution.
If the distribution has the majority of the probability in $n_2$ elements of $\X$,
we can consider the support of the optimal input distribution  
is composed of these $n_2$ elements of $\X$.
Hence, we apply Algorithm \ref{protocol1} to the case when $\X $ is the above $n_2$ elements.
That is, it is sufficient to check whether Algorithm \ref{protocol1} outputs the 
capacity only in this case.
When we employ this method, we do not need ${n_1 \choose n_2}$ repetitions.
That is, the above hybrid method works for analytical calculation.

However, 
if $C(W;x_1, \ldots, x_{n_2})$ depends on the choice of $n_2$ elements $x_1, \ldots, x_{n_2}$,
and the minimum difference
\begin{align*}
\min_
{\substack{(x_1, \ldots, x_{n_2}) \\ \neq (x_1', \ldots, x_{n_2}')}}
%{(x_1, \ldots, x_{n_2})\neq (x_1', \ldots, x_{n_2}')}
\big|C(W;x_1, \ldots, x_{n_2})-C(W;x_1', \ldots, x_{n_2}')\big|
\end{align*}
is very small, 
this idea does not work.
In this case, it is expected that
the approximately optimal input distribution
the majority of the probability in more than $n_2$ elements of $\X$.
Hence, the above method does not work.

Also, even under Condition (C'),
there is the case that the support of the optimal input distribution is composed of a smaller element than $n_2$.
In this case, even when we apply Algorithm \ref{protocol1} 
for ${n_1 \choose n_2}$ combinations, 
we cannot obtain the capacity.

When only Condition (A) holds, $Q_{Y,*}$ can be characterized as follows.
\begin{theorem}\Label{TL}
Assume that $h_{i,j}=0$ for $j=n_1+, \ldots, n_2-1$ and the parameters
$\theta^1,\ldots, \theta^{n_1-1}$ satisfies the condition \eqref{XPA}.
When the parameters
$\theta^{n_1},\ldots, \theta^{n_2-1}$ are given as
\begin{align}
&(\theta^{n_1},\ldots, \theta^{n_2-1})\nonumber \\
=&\argmin_{\eta^{n_1},\ldots, \eta^{n_2-1}}\phi(\theta^1,\ldots, \theta^{n_1-1},
 \eta^{n_1},\ldots, \eta^{n_2-1}),\Label{MIH}
\end{align}
we have $P_{\theta,Y}=Q_{Y,*}$. 
\end{theorem}

\begin{proof}
Since the objective function in \eqref{MIH}, 
the parameters $\theta^{n_1},\ldots, \theta^{n_2-1}$ achieves the minimum \eqref{MIH} if and only if
\begin{align}
\frac{\partial \phi(\theta^1,\ldots, \theta^{n_2-1})}{\partial \theta^{j}} =0 \hbox{ for } j= n_1, \ldots, n_2-1.
\Label{CPX}
\end{align}
This is because 
the concavity of $\phi$ guarantees that  there are no local minima. 
The above condition is equivalent to 
\begin{align}
\sum_{y \in \Y}f_j(y) P_{\theta,Y}(y)=0\hbox{ for } j= n_1, \ldots, n_2-1.
\end{align}
Due to \eqref{XPY}, when 
$\theta^{n_1},\ldots, \theta^{n_2-1}$ are given by \eqref{MIH},
$P_{\theta,Y}$ belongs to ${\cal M}_0$.
Due to the uniqueness by Lemma \ref{LX1},
we obtain the desired statement.
\end{proof}

Theorem \ref{TL} guarantees that 
$Q_{Y,*}$ is given as the solution of the minimization \eqref{MIH}, which is a convex minimization.
While the analytical solution of \eqref{MIH} is difficult in general, 
it is possible in the following case.
We impose the following condition for the functions
$f_{n_1}, \ldots, f_{n_2-1}$:
(I) $f_{j}(y)$ takes non-zero value only with two elements $y_j,y_j' \in \Y$ for $j=n_1+, \ldots, n_2-1$.
(II) The sets $\{y_j,y_j'\}$ for $j=n_1+, \ldots, n_2-1$ are disjoint with each other.

In this case, the relation \eqref{MIH}, i.e., \eqref{CPX}, can be simplified as
\begin{align}
0=&f_j(y_j) e^{f_j(y_j)\theta^j +\sum_{i=1}^{n_1-1}f_i(y_j) \theta^i}
\nonumber \\
&+
f_j(y_j')e^{f_j(y_j')\theta^j +\sum_{i=1}^{n_1-1}f_i(y_j') \theta^i}
\Label{XAB}
\end{align}
for $j=n_1+, \ldots, n_2-1$.
The equation \eqref{XAB} is solved as
\begin{align}
\theta^j=&\frac{1}{f_j(y_j)-f_j(y_j')}
\bigg(\sum_{i=1}^{n_1-1}(f_i(y_j') -f_i(y_j)) \theta^i \nonumber \\
&+\log \frac{-f_j(y_j')}{f_j(y_j)}\bigg)
\end{align}
for $j=n_1+, \ldots, n_2-1$.
Therefore, we can analytically calculate $P_{\theta,Y}=Q_{Y,*}$
under the conditions (I) and (II).

%Assume that $h_{i,j}=0$ for $j=n_1+, \ldots, n_2-1$ and the parameters

\section{Example}\Label{S3}
\subsection{Output system with two elements}\Label{S3-1}
First, we consider the case with $Y=\{1,2\}$.
When $\X$ and the channel $W$ satisfies 
Condition (C') in this case, 
the method described after Condition (C') works well as follows.
For any two elements $x_1\neq x_2 \in \X$,
the channel only with two inputs $x_1,x_2$ always satisfies the condition \eqref{MMR}
because $Q_{Y,*}$ is located between $W_{x_1}$ and $W_{x_2}$.
%$\Y$ is composed only of two elements.
Hence, the condition in Theorem \ref{T0} holds. 
Therefore, it is sufficient to derive a general formula for the capacity when two elements in $\X$
are fixed.

Therefore, in the following, we consider the case with $\X=\{1,2\}$ and $Y=\{1,2\}$.
%In this case, the distribution $Q_{Y,*}$ satisfying (B) needs to be located between $W_1$ and $W_2$.
We define the distributions $W_x$ for $x\in \X$ by the following vector form:
\begin{align}
W_1:= &
\left(
\begin{array}{c}
1-p \\
p\end{array}
\right),\quad
W_2:= 
\left(
\begin{array}{c}
1-q \\
q
\end{array}
\right) .
\end{align}
For simplicity, we assume that $q > p$.
We define the $2 \times 2$ matrix $V$ as 
$V:=(W_1, W_1)$.
The inverse matrix is
\begin{align}
V^{-1}=&
\frac{1}{q-p}\left(
\begin{array}{cc}
q & q-1 \\
-p & 1-p
\end{array}
\right).
\end{align}
In this case,  the parameter $\theta$ is one-dimensional and is solved to 
$h(q)-h(p)$, where $h(p)$ is the binary entropy.
Then, $\phi(\theta)$ is calculated as
\begin{align}
\phi(\theta)
=& \log \Big( e^{\frac{q(h(q)-h(p))}{q-p}}+ e^{\frac{-p(h(q)-h(p))}{q-p}}\Big)  \nonumber\\
=& \log \Big( e^{\frac{q(h(q)-h(p))}{q-p}}(1+ e^{\frac{-(q+p)(h(q)-h(p))}{q-p}})\Big)  \nonumber\\
=& \frac{q(h(q)-h(p))}{q-p} +\log \Big(1+ e^{\frac{-(q+p)(h(q)-h(p))}{q-p}}\Big) .
\end{align}
The capacity is calculated as
\begin{align}
C(W)=&\phi(\theta)-h(p)  \nonumber\\
=& \frac{p h(q)-q h(p)}{q-p} +\log \Big(1+ e^{\frac{-(q+p)(h(q)-h(p))}{q-p}}\Big) ,
\end{align}
which is a general capacity formula with 
$\X=\{1,2\}$ and $Y=\{1,2\}$.
Then, 
\begin{align}
P_{\theta,Y}=
\left(
\begin{array}{c}
e^{\frac{q(h(q)-h(p))}{q-p} -\phi(\theta)}\\
e^{\frac{-p(h(q)-h(p))}{q-p}-\phi(\theta)}
\end{array}
\right).
\end{align}
Hence, the optimal input distribution is
\begin{align}
\widehat{Q}_{X,*}=
\left(
\begin{array}{c}
\frac{1}{q-p} \Big(q -e^{\frac{-p(h(q)-h(p))}{q-p}-\phi(\theta) } \Big) \\
\frac{1}{q-p} \Big(-p +e^{\frac{-p(h(q)-h(p))}{q-p}-\phi(\theta) } \Big) 
\end{array}
\right).
\end{align}

\subsection{Output system with three elements}
\subsubsection{General problem description}
Next, we consider the case with $Y=\{1,2,3\}$.
In this case, Algorithm \ref{protocol1} does not necessarily work even under the condition (C).
Moreover, the method described after Condition (C') 
does not necessarily work even under the condition (C').
To see such a case, 
we consider the following example with
$\X=\{1,2,3,4\}$ and $\Y=\{1,2,3\}$
with $\epsilon \in [0,1/2]$.
We define the distributions $W_x$ for $x\in \X$ by the following vector form:
\begin{align}
W_1:= &
\left(
\begin{array}{c}
1-\epsilon \\
0 \\
\epsilon
\end{array}
\right),\quad
W_2:= 
\left(
\begin{array}{c}
0 \\
1 -\epsilon\\
\epsilon
\end{array}
\right) \\
W_3:= &
\left(
\begin{array}{c}
\frac{1}{2} \\
\frac{1}{2} \\
0
\end{array}
\right),\quad
W_4:= 
\left(
\begin{array}{c}
\frac{1}{2}-\epsilon \\
\frac{1}{2}-\epsilon \\
2\epsilon
\end{array}
\right) .
\end{align}
We define $3 \times 3$ matrix $V_j$ for $j\in \X$ as 
$V_1:=(W_2, W_3,W_4)$, $V_2:=(W_1, W_3,W_4)$,
$V_3:=(W_1, W_2,W_4)$, $V_4:=(W_1, W_2,W_3)$.
Their inverse matrices are
%We define the $3 \times 3$ matrix $V$ as $V:=(W_1,W_2, W_3)$. Then, the inverse matrix is calculated as
\begin{align}
V_1^{-1}=&
\left(
\begin{array}{ccc}
-\frac{1}{1-\epsilon} &\frac{1}{1-\epsilon}& 0 \\
\frac{3-2\epsilon}{2(1-\epsilon)} & \frac{1-2\epsilon}{2(1-\epsilon)} &  \frac{-1+2\epsilon}{2\epsilon}\\
\frac{1}{2(1- \epsilon)} & -\frac{1}{2(1- \epsilon)} &  \frac{1}{2\epsilon}\\
\end{array}
\right) \\
V_2^{-1}=&
\left(
\begin{array}{ccc}
\frac{1}{1-\epsilon} &-\frac{1}{1-\epsilon}& 0 \\
\frac{1-2\epsilon}{2(1-\epsilon)} & \frac{3-2\epsilon}{2(1-\epsilon)} &   \frac{-1+2\epsilon}{2\epsilon}\\
-\frac{1}{2(1- \epsilon)} & \frac{1}{2(1- \epsilon)} &  \frac{1}{2\epsilon}\end{array}
\right) \\
V_3^{-1}=&
\left(
\begin{array}{ccc}
\frac{3-2\epsilon}{2(1-\epsilon)} & \frac{1-2\epsilon}{2(1-\epsilon)} 
& -\frac{1-2\epsilon}{2\epsilon} \\
\frac{1-2\epsilon}{2(1-\epsilon)} &\frac{3-2\epsilon}{2(1-\epsilon)} 
& -\frac{1-2\epsilon}{2\epsilon} \\
-1 &-1 & \frac{1-\epsilon}{\epsilon} 
\end{array}
\right)  \\
V_4^{-1}=&
\left(
\begin{array}{ccc}
\frac{1}{2(1-\epsilon)} & -\frac{1}{2(1-\epsilon)} & \frac{1}{2\epsilon} \\
-\frac{1}{2(1-\epsilon)} & \frac{1}{2(1-\epsilon)} & \frac{1}{2\epsilon} \\
1 & 1 & \frac{-1+\epsilon}{\epsilon} 
\end{array}
\right) .
\end{align}
Also, we have 
\begin{align}
H(W_1)=&H(W_2)=h(\epsilon) \\
H(W_3)=&\log 2, \quad H(W_4)=h(2 \epsilon)+ (1-2\epsilon)\log 2.
\end{align}
When we apply Algorithm \ref{protocol1} to the three components in $V_j$,
we denote $\theta$, $\phi(\theta)$,
$P_{\theta,Y}$, $\widehat{Q}_{X,*}$, and $\phi(\theta)-H(W_n)$ by $\theta_j$, $\phi_j(\theta_j)$, $P_{j,Y}$ and $\widehat{Q}_{j,X}$, and $C_j$,
respectively.
In the following, we discuss our model dependently of the value of $j$.

\subsubsection{Case that $j=4$}
First, we consider the case that $j=4$, i.e., 
the channel is composed of three inputs $\{1,2,3\}$.
Then, we have
\begin{align}
\theta_4=
\left(
\begin{array}{c}
h_{4,\epsilon}\\
h_{4,\epsilon}
\end{array}
\right) 
\end{align}
with $h_{4,\epsilon}:=\log 2- h(\epsilon) $
and
\begin{align}
\phi_4(\theta_4)=&
\log \Big(2+ e^{ \frac{h_{4,\epsilon}}{\epsilon} }\Big) .
%\log (\frac{2^{\frac{3}{2}}}{\epsilon^{\epsilon} (\frac{1}{2}-\epsilon)^{(\frac{1}{2}-\epsilon)}}+1) .
\end{align}
Thus,
\begin{align}
&C_4= \phi_4(\theta_4)-\log 2 
= \log \bigg(1+\frac{e^{\frac{h_{4,\epsilon}}{\epsilon}}}{2}\bigg)
\Label{NNO} 
\end{align}
Hence,
\begin{align}
P_{4,Y}=
\left(
\begin{array}{c}
e^{- \phi_4(\theta_4)} \\
e^{- \phi_4(\theta_4)} \\
e^{\frac{h_{4,\epsilon}}{\epsilon}- \phi_4(\theta_4)} 
\end{array}
\right) .
\end{align}
Therefore, 
\begin{align}
&\widehat{Q}_{4,X}
=V_4^{-1}P_{4,Y} \nonumber\\
=&
\left(
\begin{array}{c}
\frac{1}{2\epsilon} e^{\frac{h_{4,\epsilon}}{\epsilon} - \phi_4(\theta_4)}
\\
\frac{1}{2\epsilon} e^{\frac{h_{4,\epsilon}}{\epsilon} - \phi_4(\theta_4)}
\\
2e^{- \phi_4(\theta_4)}
-\frac{1-\epsilon }{\epsilon}e^{\frac{h_{4,\epsilon}}{\epsilon} - \phi_4(\theta_4)}
\end{array}
\right) .
\end{align}
While the first and second components of $\widehat{Q}_{4,X}$ are always positive value,
the third component has a possibility to have a negative value.
The non-negativity of the first component is equivalent to 
the following condition:
\begin{align}
1\ge g_1(\epsilon),\Label{AA}
\end{align}
where $g_1(\epsilon):=\frac{1-\epsilon}{2 \epsilon} e^{\frac{h_{4,\epsilon}}{\epsilon}}$.
In \eqref{AA}, 
the first inequality corresponds to the non-negativity of the third component and
the second inequality corresponds to the non-negativity of the first component.
Fig. \ref{g-graph} numerically plots the function $g_1(\epsilon)$.
It shows that $\widehat{Q}_{4,X}$ is a probability distribution when 
$0.3588 \le \epsilon $.
That is, $C_4$ is achievable for $0.3588 \le \epsilon $.
When 
$\epsilon < 0.3588$,
the third component of $\widehat{Q}_{4,X}$ is negative.
Hence, $C_4$ is not achievable.
Due to Theorem \ref{T-1},
the optimal input distribution in this case
has the support in $\{1,2\}$.
In this case, due to the symmetry, the uniform distribution on 
$\{1,2\}$ is optimal.
That is, the capacity with the input set $\{1,2,3\}$ is
\begin{align}
{C}_{*}:=-(1-\epsilon)\log \frac{1-\epsilon}{2}- \epsilon \log \epsilon
-  h(\epsilon)= (1-\epsilon)\log 2. \Label{Cs}
\end{align}
 
\begin{figure}[htbp]
%\centering
%\includegraphics[scale=0.4]{MHRepeater.png}
\begin{center}
  \includegraphics[width=1\linewidth]{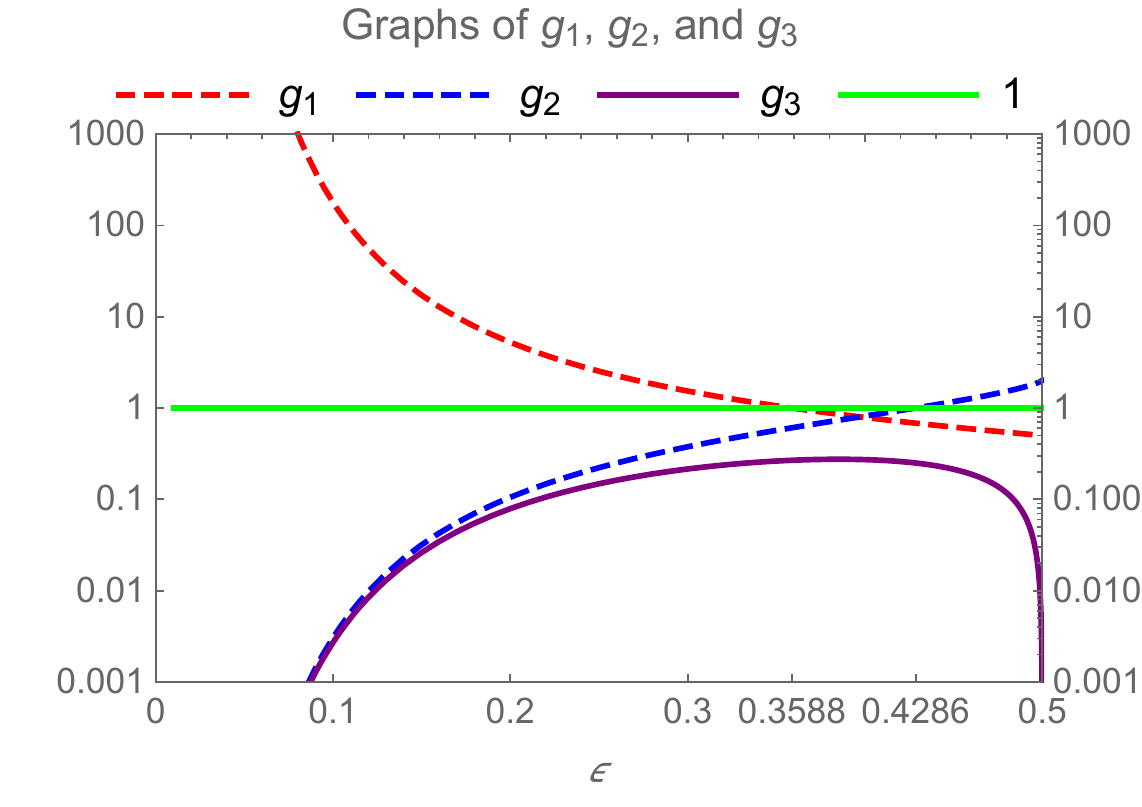}
  \end{center}
\caption{Graphs of functions $g_1,g_2,$ and $g_3$ with logarithmic scale.
Red dashed curve expresses $g_1$. 
Blue dashed curve expresses $g_2$. 
Purple solid curve expresses $g_3$. 
Green solid line expresses $1$. 
Red dashed curve $g_1$ and Blue dashed curve $g_2$
across Green solid line at $0.3588$ and $0.4286$, respectively.
}
\Label{g-graph}
\end{figure}

\subsubsection{Case that $j=3$}
We consider the case that $j=3$, i.e., 
the channel is composed of three inputs $\{1,2,4\}$.
Then, we have
\begin{align}
\theta_3=
\left(
\begin{array}{c}
h_{3,\epsilon}\\
h_{3,\epsilon}
\end{array}
\right) 
\end{align}
with $h_{3,\epsilon}:=h(2 \epsilon)+ (1-2\epsilon)\log 2- h(\epsilon) $
and
\begin{align}
\phi_3(\theta_3)=&
\log \Big(2e^{2 h_{3,\epsilon}}+ e^{ -\frac{(1-2\epsilon)h_{3,\epsilon}}{\epsilon} }\Big)  \nonumber\\
=&\log \Big(2+ e^{ -\frac{h_{3,\epsilon}}{\epsilon} }\Big) +
2 h_{3,\epsilon}.
%\log (\frac{2^{\frac{3}{2}}}{\epsilon^{\epsilon} (\frac{1}{2}-\epsilon)^{(\frac{1}{2}-\epsilon)}}+1) .
\end{align}
Thus,
\begin{align}
&C_3= \phi_3(\theta_3)-h(2 \epsilon)- (1-2\epsilon)\log 2  \nonumber\\
=&\log \Big(2+ e^{ -\frac{ h_{3,\epsilon}}{\epsilon} }\Big) +
2h_{3,\epsilon}
-h(2 \epsilon)- (1-2\epsilon)\log 2  \nonumber\\
=&\log \Big(2+ e^{ -\frac{h_{3,\epsilon}}{\epsilon} }\Big) +
h(2 \epsilon)+ (1-2\epsilon)\log 2-2 h( \epsilon) .
\Label{NNO3} 
\end{align}
Hence,
\begin{align}
P_{3,Y}=
\frac{1}{2+ e^{ -\frac{h_{3,\epsilon}}{\epsilon} }}
\left(
\begin{array}{c}
1 \\
1\\
e^{-\frac{h_{3,\epsilon}}{\epsilon}}
\end{array}
\right) .
\end{align}
Therefore, 
\begin{align}
&\widehat{Q}_{3,X}
=V_3^{-1}P_{3,Y} \nonumber\\
=&
\frac{1}{2+ e^{ -\frac{h_{3,\epsilon}}{\epsilon} }}
\left(
\begin{array}{c}
1 -
\frac{1-2\epsilon}{2\epsilon} 
e^{-\frac{h_{3,\epsilon}}{\epsilon} } 
\\
1-
\frac{1-2\epsilon}{2\epsilon} 
e^{-\frac{h_{3,\epsilon}}{\epsilon}} 
\\
-2 +
\frac{1-\epsilon}{\epsilon}e^{-\frac{ h_{3,\epsilon}}{\epsilon}} 
\end{array}
\right) .
\end{align}
While the first and second components of $\widehat{Q}_{3,X}$ always have positive values,
the third component has a possibility to have a negative value.
The non-negativity of the first component is equivalent to 
the following condition:
\begin{align}
g_2(\epsilon) \ge 1\ge g_3(\epsilon),
\Label{AA3}
\end{align}
where
$g_2(\epsilon):=
\frac{1-\epsilon}{2 \epsilon} e^{-\frac{h_{3,\epsilon}}{\epsilon}}$
and
$g_3(\epsilon):=
\frac{1-2\epsilon}{2 \epsilon} e^{-\frac{h_{3,\epsilon}}{\epsilon}}$.
In \eqref{AA3}, 
the first inequality corresponds to the non-negativity of the third component and
the second inequality corresponds to the non-negativity of the first component.
However,  as numerically plotted in Fig. \ref{g-graph}, 
$g_3(\epsilon)\le 1$ for $\epsilon<\frac{1}{2}$ and 
$g_2(\epsilon)< 1$ for $\epsilon<0.4286$.
Hence, when $\epsilon \ge 0.4286$, $C_3$ is achievable, i.e., it gives the capacity under the case $j=3$.

When $\epsilon < 0.4286$,
the third component of $\widehat{Q}_{3,X}$ is negative.
In this case, due to Theorem \ref{T-1},
the optimal distribution has support $\{1,2\}$.
Hence, the capacity with the input set $\{1,2,4\}$ 
is $C_*$ defined in \eqref{Cs}.

\subsubsection{Case that $j=1$}
Next, we consider the case that $j=1$, i.e., 
the channel is composed of three inputs $\{2,3,4\}$.
Then, we have
\begin{align}
\theta_1=
\left(
\begin{array}{c}
h_{3,\epsilon}\\
h_{1,\epsilon}
\end{array}
\right) 
\end{align}
with $h_{1,\epsilon}:=h(2 \epsilon)-2\epsilon \log 2$
and
\begin{align}
&\phi_1(\theta_1)\nonumber\\
=&
\log \Big(
e^{-\frac{1}{1-\epsilon} h_{3,\epsilon}+(\frac{1}{2(1-\epsilon)}+1) h_{1,\epsilon}}  \nonumber\\
&+
e^{\frac{1}{1-\epsilon} h_{3,\epsilon}+(-\frac{1}{2(1-\epsilon)}+1) h_{1,\epsilon}}
+ e^{ (-\frac{1}{2\epsilon}+1) h_{1,\epsilon} }\Big)  \nonumber\\
=& \log \Big(
\frac{1}{1-\epsilon} 
\Big(\frac{1-2\epsilon}{4}\Big)^{\frac{1-2\epsilon}{2-2\epsilon}}\nonumber \\
&+(1-\epsilon) \Big(\frac{1-2\epsilon}{4}\Big)^{-\frac{1-2\epsilon}{2-2\epsilon}}
+4 \epsilon  (1-2 \epsilon )^{\frac{1-2 \epsilon }{2\epsilon}}
\Big)+h_{1,\epsilon}.
%\log (\frac{2^{\frac{3}{2}}}{\epsilon^{\epsilon} (\frac{1}{2}-\epsilon)^{(\frac{1}{2}-\epsilon)}}+1) .
\end{align}
Thus,
\begin{align}
&C_1= \phi_1(\theta_1)-h(2 \epsilon)- (1-2\epsilon)\log 2 \nonumber\\
=&\log \Big(
\frac{1}{1-\epsilon} 
\Big(\frac{1-2\epsilon}{4}\Big)^{\frac{1-2\epsilon}{2-2\epsilon}}
+(1-\epsilon) \Big(\frac{1-2\epsilon}{4}\Big)^{-\frac{1-2\epsilon}{2-2\epsilon}}\nonumber\\
&
+4 \epsilon  (1-2 \epsilon )^{\frac{1-2 \epsilon }{2\epsilon}}
\Big)-\log 2.
% \log (2+e^{\frac{h_\epsilon}{\epsilon})})-h(\frac{1}{2}-\epsilon)
\Label{NNO2} 
\end{align}
Since 
\begin{align}
e^{ -\frac{1}{2\epsilon}h_{1,\epsilon} } 
=4 \epsilon  (1-2 \epsilon )^{\frac{1-2 \epsilon }{2\epsilon}},
%4 \epsilon e^{\frac{1-2 \epsilon }{2\epsilon}\log (1-2 \epsilon )} ,
\end{align}
we have
\begin{align}
P_{1,Y}=&
\left(
\begin{array}{c}
e^{-\frac{1}{1-\epsilon} h_{3,\epsilon}+(\frac{1}{2(1-\epsilon)}+1) h_{1,\epsilon}- \phi_1(\theta_1)} \\
e^{\frac{1}{1-\epsilon} h_{3,\epsilon}+(-\frac{1}{2(1-\epsilon)}+1) h_{1,\epsilon}- \phi_1(\theta_1)} \\
e^{ (-\frac{1}{2\epsilon}+1) h_{1,\epsilon} - \phi_1(\theta_1)} 
\end{array}
\right)  \\
=&
e^{ h_{1,\epsilon} - \phi_1(\theta_1)} 
\left(
\begin{array}{c}
\frac{1}{1-\epsilon} 
(\frac{1-2\epsilon}{4})^{\frac{1-2\epsilon}{2-2\epsilon}}
 \nonumber \\
(1-\epsilon) (\frac{1-2\epsilon}{4})^{-\frac{1-2\epsilon}{2-2\epsilon}}
 \\
4 \epsilon  (1-2 \epsilon )^{\frac{1-2 \epsilon }{2\epsilon}}
\end{array}
\right)  .
\end{align}
Therefore, 
\begin{align}
\widehat{Q}_{1,X}
=V_1^{-1}P_{1,Y} %\nonumber\\
=
e^{ h_{1,\epsilon} - \phi_1(\theta_1)} 
\left(
\begin{array}{c}
\kappa_1
\\
\kappa_2
\\
\kappa_3
%e^{ -\frac{1}{2\epsilon}h_{1,\epsilon} } 
\end{array}
\right) ,
\end{align}
where
\begin{align*}
\kappa_1:=&
-\frac{1}{(1-\epsilon)^2} 
(\frac{1-2\epsilon}{4})^{\frac{1-2\epsilon}{2-2\epsilon}}
+(\frac{1-2\epsilon}{4})^{-\frac{1-2\epsilon}{2-2\epsilon}} \\
\kappa_2:=&
\frac{3-2\epsilon}{2(1-\epsilon)^2} 
(\frac{1-2\epsilon}{4})^{\frac{1-2\epsilon}{2-2\epsilon}}   
+\frac{1-2\epsilon}{2}
(\frac{1-2\epsilon}{4})^{-\frac{1-2\epsilon}{2-2\epsilon}}\\
&+ (-2+4\epsilon) (1-2 \epsilon )^{\frac{1-2 \epsilon }{2\epsilon}}
%e^{ -\frac{1}{2\epsilon}h_{1,\epsilon} } 
\\
\kappa_3:=&
\frac{1}{2(1- \epsilon)^2}
(\frac{1-2\epsilon}{4})^{\frac{1-2\epsilon}{2-2\epsilon}}
-\frac{1}{2}(\frac{1-2\epsilon}{4})^{-\frac{1-2\epsilon}{2-2\epsilon}}\\
&+2 (1-2 \epsilon )^{\frac{1-2 \epsilon }{2\epsilon}}.
\end{align*}
Taking the limit $\epsilon \to 0$, we have
\begin{align}
\lim_{\epsilon \to 0}\widehat{Q}_{1,X}
=
\frac{2}{5}\left(
\begin{array}{c}
\frac{3}{2} \\
\frac{7}{4}-\frac{2}{e}\\
\frac{2}{e}-\frac{3}{4}
\end{array}
\right) 
=
\left(
\begin{array}{c}
\frac{3}{5} \\
\frac{7}{10}-\frac{4}{5 e}\\
\frac{4}{5 e}-\frac{3}{10}
\end{array}
\right) .
\end{align}

\begin{figure}[htbp]
%\centering
%\includegraphics[scale=0.4]{MHRepeater.png}
\begin{center}
  \includegraphics[width=1\linewidth]{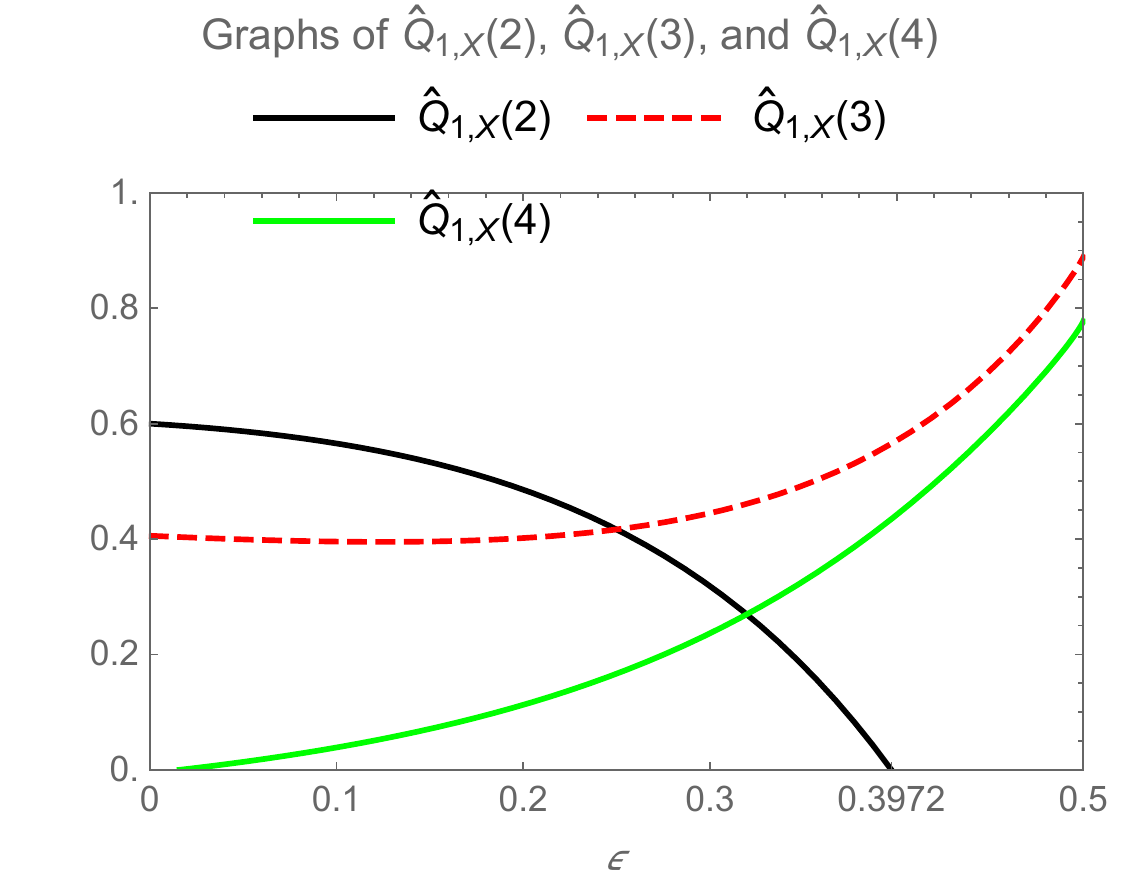}
  \end{center}
\caption{Graph of the function $\widehat{Q}_{1,X}$.
Black solid curve expresses $\widehat{Q}_{1,X}(2)$. 
Red dashed curve expresses $\widehat{Q}_{1,X}(3)$. 
Green solid line expresses $\widehat{Q}_{1,X}(4)$. 
The values $\widehat{Q}_{1,X}(3)$ and $\widehat{Q}_{1,X}(4)$ are always positive.
The value $\widehat{Q}_{1,X}(2)$ is positive only when $\epsilon \le 0.3972$.
}
\Label{GraphQ1}
\end{figure}   

Fig. \ref{GraphQ1} shows numerical plots of $\widehat{Q}_{1,X}(2),$
$\widehat{Q}_{1,X}(3),$ and $\widehat{Q}_{1,X}(4)$. 
Although $\widehat{Q}_{1,X}(3)$ and $\widehat{Q}_{1,X}(4)$ are always positive,
$\widehat{Q}_{1,X}(2)$ is positive only for $\epsilon \ge 0.3972$.
Hence, when $\epsilon < 0.3972$, 
due to Theorem \ref{T-1}, the capacity of case $j=1$ equals the capacity of the channel with inputs $3$ and $4$.
In this case, 
we cannot use Algorithm \ref{protocol1} because the size of the input system is smaller than the size of the output system.
%Also, we cannot use the result of Subsection \ref{S3-1} because the number of the output system is 3.
Assume that $P_X(3)=1-p$ and $P_X(4)=p$.
The mutual information between $X$ and $Y$ is
\begin{align}
&h(2\epsilon p)-p h(2\epsilon)\nonumber\\
=&(1-2 \epsilon p )\log 2+ h(2\epsilon p)\nonumber\\
&-(1-p)\log 2 -p (h(2\epsilon) +(1-2 \epsilon) \log 2 )  .
\end{align}
Then, the maximum mutual information is achieved when 
\begin{align*}
p=&
\frac{1}{2\epsilon\big(1+e^{\frac{h(2\epsilon)}{2\epsilon}}\big)} %\nonumber \\
=\frac{1}{2\epsilon
(1+(1-2\epsilon)^{-\frac{1-2\epsilon}{2\epsilon}} (2\epsilon)^{-1})} \nonumber \\
=&\frac{1}{
2\epsilon+(1-2\epsilon)^{-\frac{1-2\epsilon}{2\epsilon}} }.
\end{align*}

The capacity of this case is
\begin{align}
&C_{**}:= h\Big(\frac{2\epsilon}{
2\epsilon+(1-2\epsilon)^{-\frac{1-2\epsilon}{2\epsilon}} }
\Big)-\frac{h(2\epsilon)}{
2\epsilon+(1-2\epsilon)^{-\frac{1-2\epsilon}{2\epsilon}} }  \nonumber\\
=& \frac{2\epsilon}{
2\epsilon+(1-2\epsilon)^{-\frac{1-2\epsilon}{2\epsilon}} }
\log \Big(2\epsilon+(1-2\epsilon)^{-\frac{1-2\epsilon}{2\epsilon}} \Big) \nonumber \\
&- 
\Big(\frac{(1-2\epsilon)^{-\frac{1-2\epsilon}{2\epsilon}}}{2\epsilon+(1-2\epsilon)^{-\frac{1-2\epsilon}{2\epsilon}} }\Big) 
\log 
\Big(\frac{(1-2\epsilon)^{-\frac{1-2\epsilon}{2\epsilon}}}{2\epsilon+(1-2\epsilon)^{-\frac{1-2\epsilon}{2\epsilon}} }\Big) \nonumber \\
&- \frac{1-2\epsilon}{2\epsilon+(1-2\epsilon)^{-\frac{1-2\epsilon}{2\epsilon}} } 
\log (1-2\epsilon).
\end{align}
Due to the symmetry, we can discuss the case with $j=2$.

\subsubsection{Derivation of $C(W)$}
Based on the above discussion, we discuss the capacity of the channel $W$ with input system $\X=\{1,2,3,4\}$.
When $0 \le \epsilon \le 0.3588 $, 
$C_{1}$ is the capacity for the case $j=1$, and
$C_*$ is the capacity for the cases $j=3,4$.
Since $C_*\ge C_{1}$ in this case, $C_*$ is the capacity of the channel $W$.

When $0.3588 < \epsilon \le 0.3972$,
$C_{1}$ is the capacity for the case $j=1$, 
$C_*$ is the capacity for the cases $j=3$, and 
$C_4$ is the capacity for the cases $j=4$.
Since $C_4\ge C_*,C_{1}$ in this case, $C_4$ is the capacity of the channel $W$.

In fact, as seen in Fig. \ref{Capacity-en}, 4 curves $C_1$, $C_3$, $C_4$, and $C_{**}$
intersect at $0.3972$.
For $0.3972< \epsilon \le \frac{1}{2}$, 
$C_{**}$ is the capacity for the case $j=1$, 
$C_3$ or $C_*$ is the capacity for the cases $j=3$, and 
$C_4$ is the capacity for the cases $j=4$.
Since $C_{**}\ge C_*,C_3,C_4$ in this case, $C_{**}$ is the capacity of the channel $W$.
Overall, the capacity $C(W)$ of the channel $W$ is calculated as follows.
\begin{align}
C(W)=
\left\{
\begin{array}{ll}
C_* & \hbox{ when } 0 \le \epsilon \le 0.3588 \\ 
C_4 & \hbox{ when } 0.3588 < \epsilon \le 0.3972 \\ 
C_{**} & \hbox{ when } 0.3972 < \epsilon \le 1/2.
\end{array}
\right.
\end{align}

\begin{figure}[htbp]
%\centering
%\includegraphics[scale=0.4]{MHRepeater.png}
\begin{center}
  \includegraphics[width=1\linewidth]{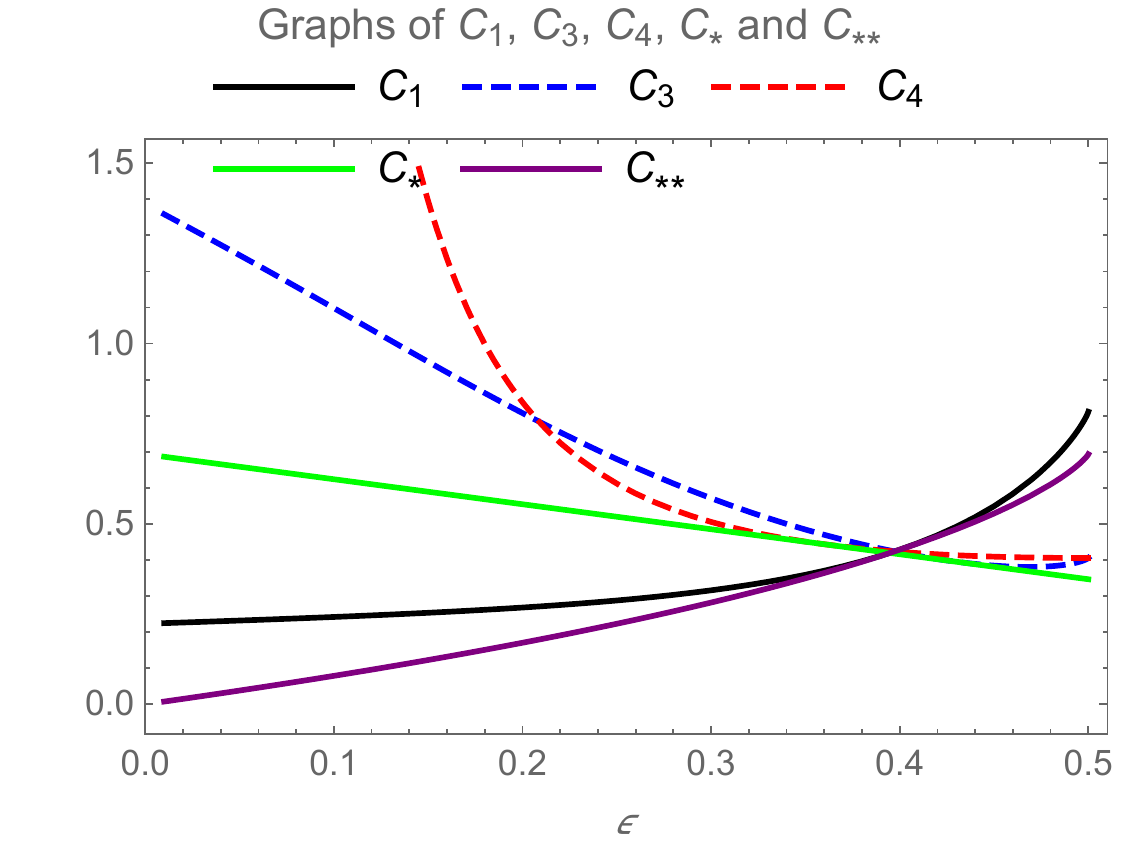}
  \end{center}
\caption{Graphs of functions $C_1,C_3,C_4,C_{*}$ and $C_{**}$.
Black solid curve expresses $C_1$. 
Blue dashed curve expresses $C_3$. 
Red dashed curve expresses $C_4$. 
Green solid line expresses $C_*$. 
Purple solid curve expresses $C_{**}$. 
Its enlarged view is given as Fig. \ref{Capacity-en}.
}
\Label{Capacity-c}
\end{figure}

\begin{figure}[htbp]
%\centering
%\includegraphics[scale=0.4]{MHRepeater.png}
\begin{center}
  \includegraphics[width=1\linewidth]{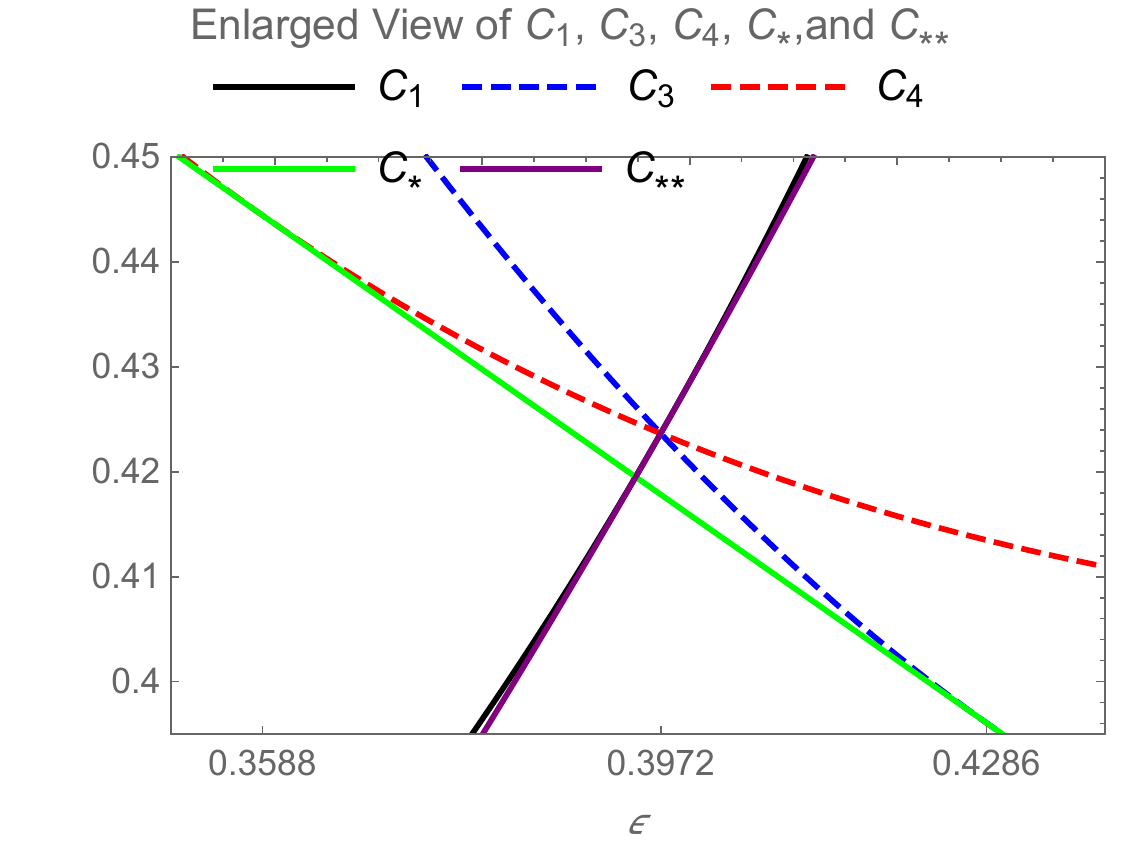}
  \end{center}
\caption{Enlarged view of graphs of functions $C_1,C_3,C_4,C_{*}$ and $C_{**}$.
The explanations for 5 curves are the same as Fig. \ref{Capacity-c}.
4 curves $C_1$, $C_3$, $C_4$, and $C_{**}$ intersect at $0.3972$.
In particular, $C_1$ touches $C_{**}$ at $0.3972$
$C_3$ and $C_4$ touch $C_{*}$ at $0.3588$ and $0.4286$, respectively.
That is, the inequalities $C_1\ge C_{**}$ and $C_3,C_4 \ge C_{*}$ hold always.
}
\Label{Capacity-en}
\end{figure}

\section{Capacity of classical-quantum channel}\Label{S12}
%Next, we explain how to apply our method to various quantum settings.
%\subsection{Capacity of classical-quantum channel}
%The most typical example in the quantum setting.
Next, we discuss a classical-quantum channel from the classical system ${\cal X}:=\{1,\ldots, n_1\}$ to the quantum system
${\cal H}$ with dimension $n_2$, which is given as a set of density matrices $\{W_j\}_{j=1}^{n_1}$.
We denote the set of density matrices on ${\cal H}$ by ${\cal S}({\cal H})$.
For density matrices $\rho,\sigma \in {\cal S}({\cal H})$, 
the entropy $H(\rho)$ and the divergence $D(\rho\|\sigma)$ are defined as
\begin{align}
H(\rho):= -\Tr \rho \log \rho, \quad
D(\rho\|\sigma):= \Tr \rho (\log \rho - \log \sigma). \Label{AAB}
\end{align}

Under this classical-quantum channel,
the capacity of classical-quantum channel $W=\{W_j\}_{j=1}^{n_1}$ 
is defined as \cite{Ho1,Ho2,Ho3,SW1}
\begin{align}
C_q(W):=\max_{P \in {\cal P}_{{\cal X}}} \sum_{x \in \X} P(x)  D\bigg( W_x \bigg\| \sum_{x'\in \X} P(x') W_{x'}\bigg),
%\sum_{j=1}^{n_1} p_j D \bigg(W_j\bigg\| \sum_{j'=1}^k p_{j'} W_{j'}\bigg).
\end{align}
%where
%given an input probability distribution $P \in {\cal P}_{{\cal X}}$ on the classical system ${\cal X}$,
The capacity of classical-quantum channel 
has the following form \cite{OP,SW2}
\begin{align}
C_q(W)
= \min_{\sigma \in {\cal S}({\cal H}) }\max_{x \in \X} D( W_x\| \sigma) ,\Label{MOA2}
\end{align}

Statements similar to statements in Section \ref{S10} can be shown in this case of cq-channel
by using quantum information geometry based on Kubo-Mori-Bogoliubov Fisher information, which is directly linked to quantum relative entropy \eqref{AAB} \cite{Amari-Nagaoka}, \cite[Chapter 7]{Springer}.
Here, for the calculation of $C_q(W)$, 
we consider only the algorithm corresponding to Algorithm \ref{protocol1}.
Hence, we consider the case under the following condition similar to Condition (C).   
\begin{description}
\item[(D)]
$n_1=n_2^2$ and $W_1, \ldots, W_{n_2^2}$ are linearly independent.
\end{description}

%Since the dimension of the convex full of $\{W_x\}_{x \in \X}$ equals the dimension of ${\cal S}({\cal H})$,
%due to Condition (D), 

%Hence, when a density matrix $\sigma$ satisfies (E), $ D( W_x\| \sigma) $ realizes the capacity.

We choose $n^2-1$ linearly independent Hermitian matrices $A_1, \ldots, A_{n_2^2-1}$ on ${\cal H}$ such that
\begin{align}
\Tr W_{n_2^2} A_j =0 \Label{CO2}
\end{align}
for $j=1, \ldots, n_2^2-1$.
We define the matrix $(h_{i,j})$
\begin{align}
h_{i,j}:= \Tr W_i A_j.\Label{XPA2}
\end{align}
Due to Condition (D),
%there exists an element $x_0 \in \X$ such that
the $n_2^2-1$ vectors $\{(h_{x,j})_{j=1}^{n_2^2-1} \}_{x =1}^{n_2^2-1}$ 
are linearly independent.

Given an $n_2^2-1$-dimensional parameter $\theta=(\theta^1, \ldots, \theta^{n_2^2-1})$,
we define the density matrix $\rho_{\theta}$ as
\begin{align}
\rho_{\theta}=\exp \Bigg( \sum_{j=1}^{n_2^2-1} A_j \theta^j-\phi(\theta)  \Bigg),\Label{nat2}
\end{align}
where 
\begin{align}
\phi(\theta):= \log \Tr \exp \Bigg( \sum_{j=1}^{n_2^2-1} A_j \theta^j \Bigg)  \Label{XZP2}.
\end{align}

We have the following theorem.
\begin{theorem}\Label{T2}
Assume that the parameters $\theta^1, \ldots, \theta^{n_2^2-1}$ satisfies
the condition
\begin{align}
\sum_{j=1}^{n_2^2-1}h_{i,j} \theta^j= -H(W_i)+H(W_{n_2^2})\Label{MXP2}
\end{align}
for $i=1, \ldots, n_2^2-1$
Then, we have
\begin{align}
D(W_x\| \rho_{\theta})=\phi(\theta)-H(W_{n_2^2}) %=C_q(W).
\end{align}
for any element $x\in {\cal A}$.
\end{theorem}

\begin{proof}
The condition \eqref{MXP2} implies that
\begin{align}
\Tr W_i \sum_{j=1}^{n_2^2-1} A_j \theta^j
=\sum_{j=1}^{n_2^2-1}h_{i,j} \theta^j
=-H(W_i)+H(W_{n_2^2}).
\end{align}

For any element $x (\neq n_2^2) \in \X$, we have
\begin{align}
& D(W_x\| \rho_{\theta})=
\Tr W_x  (\log W_x -\log \rho_{\theta} ) \nonumber \\
=&-H(W_x)
-\Tr W_x  \Bigg(
\sum_{j=1}^{n_2^2-1} A_j \theta^j-\phi(\theta) \Bigg) \nonumber \\
=&-H(W_x)- \big(-H(W_x) +H(W_{n_2^2}) -\phi(\theta) \big)\nonumber \\
=& \phi(\theta)-H(W_{n_2^2}).\Label{CM1}
\end{align}
Also, we have
\begin{align}
& D(W_{n_2^2}\| \rho_{\theta})=
\Tr W_{n_2^2}  (\log W_{n_2^2} -\log \rho_{\theta} ) \nonumber \\
=&-H(W_{n_2^2})
-\Tr W_{n_2^2}  \Bigg(
\sum_{j=1}^{n_2^2-1} A_j \theta^j-\phi(\theta)\Bigg) \nonumber \\
=&-H(W_{n_2^2})-(-\phi(\theta))
=\phi(\theta)-H(W_{n_2^2}).\Label{CM2}
\end{align}
The combination of \eqref{CM1} and \eqref{CM2} implies the desired statement.
%Hence, when $ D( W_x\| \rho_{\theta}) $ does not depend on $x \in \X$,
%it realizes the capacity.
\end{proof}

Then, we consider the following condition:
\begin{description}
\item[(E)]
$ D( W_x\| \sigma) $ does not depend on $x \in \X$.
\end{description}

\begin{lemma}\Label{LL4}
When Condition (D) holds, 
only one density matrix $\sigma$ on ${\cal H}$ satisfies the condition (E).
We denote such a density matrix by $\sigma_*$.
\end{lemma}
%Lemma \ref{LL4} is shown in Appendix \ref{ALL4}.
\begin{proof}
\if0
We choose the parameters $\eta_{j,x}$ as
\begin{align}
\eta_{j,x}:= \Tr W_x A_j
\end{align}
for $j=1, \ldots, n_2^2-1$.
\fi

We have
\begin{align}
D( W_x\| \rho_{\theta})= -H(W_x) -\sum_{j=1}^{n_2^2-1} \theta^j h_{x,j} 
-\phi(\theta).
\end{align}
Condition (E) with $\sigma=\rho_{\theta}$ is rewritten as
\begin{align}
&-H(W_x) -\sum_{j=1}^{n_2^2-1} \theta^j h_{x,j} -\phi(\theta)\nonumber \\
=&-H(W_{n_2^2}) -\sum_{j=1}^{n_2^2-1} \theta^j h_{n_2^2} -\phi(\theta)
=-H(W_{n_2^2}) -\phi(\theta)
\end{align}
for $x=1, \ldots, n_2^2-1$, where
the final equation follows from \eqref{CO2}.
This condition is rewritten as
\begin{align}
-H(W_x)+H(W_{n_2^2})
= -\sum_{j=1}^{n_2^2-1}h_{x,j} \theta^j  \Label{SASC}
\end{align}
for $x=1, \ldots, n_2^2-1$.
Since the matrix $h_{i,j}$ is invertible,
only one vector $\theta=(\theta^1, \ldots, \theta^{n_2^2-1})  $ 
satisfies \eqref{SASC}, i.e., Condition (E).
\end{proof}

The relation \eqref{MOA2} guarantees that
\begin{align}
C_q(W) \le D( W_x\| \sigma_*)\Label{CA}
\end{align}
for any element $x \in \X$.

Due to Theorem \ref{T2}, when $\theta$ satisfies the condition \eqref{MXP2},
the density matrix $\rho_{\theta}$ equals $\sigma_*$.
To construct our algorithm, we add the $n_2^2$-th Hermitian matrix $A_{n_2^2}$
and define $h_{i,j}$ by \eqref{XPA2} for $i,j=1, \ldots, n_2^2$.
To find the input distribution $\widehat{Q}_{X,*} $ to achieve the maximum 
\eqref{MOA2}, 
we consider the equation $ \sum_x W(y|x)\widehat{Q}_{X,*}(x)= \rho_{\theta}$, which can be rewritten as
\begin{align}
\sum_{x\in \X}\widehat{Q}_{X,*}(x)  h_{x,j}
\bigg( = \sum_{x\in \X}\widehat{Q}_{X,*}(x)  \Tr W_x A_j\bigg) &= \Tr \rho_{\theta} A_j . \Label{AB}
\end{align}
If we have 
\begin{align}
\widehat{Q}_{X,*}(x)\ge 0 \hbox{ for }x \in \X,
 \Label{MMR2}
\end{align}
since Lemma \ref{LL4} guarantees that $\rho_{\theta}=\sigma_*$,
\eqref{CA} guarantees that
\begin{align}
D(W_x\| \rho_\theta)=C_q(W), \Label{CXP}
\end{align}
i.e., the solution gives the capacity.

\if0
To find the input distribution $Q_X $ to achieve the maximum 
\eqref{MOA2} from the above miximizer $\rho_{\theta}$, 
we consider the following simultaneous equation
\begin{align}
\begin{aligned}
\sum_{x}Q_X(x)  h_{x,j}
\Big( = \sum_{x}Q_X(x)  \Tr W_x A_j\Big) &= \Tr \rho_{\theta} A_j \\
\sum_{x}Q_X(x) &=1
\end{aligned}\Label{AB}
\end{align}
for $j=1, \ldots,n_2^2-1 $.
Since the number of equation equals $n_2^2$, i.e., $n$,
\eqref{AB} uniquely determines the optimal input distribution $Q_X$.
\fi
Therefore, in the same way as Algorithm \ref{protocol1}, 
we propose Algorithm \ref{protocol5} based on Theorem \ref{T2} and Lemma \ref{LL4}.
%can simply compute $C_q(W)$ by 

Now, we describe two Hermitian matrices $X,Y$on $\H$ 
by two $n^2$-dimensional vectors $x=(x_j)_{j=1}^{n_2^2}$ and $y=(y_j)_{j=1}^{n_2^2}$ 
%by focusing on each matrix component 
as follows.
\begin{align}
&X\nonumber \\
=& \sum_{j=1}^{n_2} x_j|j\rangle \langle j| \nonumber \\
&+ \sum_{j=1}^{n_2-1} \sum_{j'=1}^{j-1}
 \frac{x_{n_2+  j(j-1)/2+j' }}{\sqrt{2}}
(|j\rangle \langle j'|+|j'\rangle \langle j|)\nonumber \\
&+ \sum_{j=1}^{n_2-1} \sum_{j'=1}^{j-1}
 \frac{x_{n_2(n_2+1)/2+  j(j-1)/2+j' }}{\sqrt{2}}
(i |j\rangle \langle j'|-i |j'\rangle \langle j|).
\end{align}
Here, the matrix $Y$ is defined in the same way by using $y=(y_j)_{j=1}^{n_2^2}$. 
Then, we have 
\begin{align}
\Tr XY=
\sum_{j=1}^{n_2^2} x_j y_j.
\end{align}
In this sense, 
$(W_1, \ldots, W_{n_2^2})$
and 
$(A_1, \ldots, A_{n_2^2})$ can be considered as $n_2^2\times n_2^2$ matrices.
Then, Step 1 of Algorithm \ref{protocol5} can be done by calculating the inverse matrix of the matrix corresponding to 
$(W_1, \ldots, W_{n_2^2})$.
Hence,  
Step 1 has calculation complexity $O(n_2^6)$.
In Step 2, the calculation of all of $H(W_i)$ needs 
calculation complexity $O(n_1 n_2^3)=O(n_2^5)$.
Hence, Step 2 has calculation complexity $O(n_2^5)$ in total.
In Step 3, 
the calculation of $\sum_{j=1}^{n_2^2-1} A_j \theta^j$ has calculation complexity $O(n_2^4)$.
and the calculation of $ \exp \Big( \sum_{j=1}^{n_2^2-1} A_j \theta^j \Big)$
and its trace 
 has calculation complexity $O(n_2^3)$.
Step 3 has calculation complexity $O(n_2^4)$ in total.
In Step 5, the calculation of all of $\Tr \rho_\theta A_j$
 has calculation complexity $O(n_1 n_2^2)=O(n_2^4)$
since $ \exp \Big( \sum_{j=1}^{n_2^2-1} A_j \theta^j \Big)$
and its trace are already calculated.
Hence, the total calculation complexity is $O(n_2^6)$.
\begin{algorithm}
\caption{Exact algorithm for classical channel capacity}
\Label{protocol5}
\begin{algorithmic}
\STATE {Step 1: Choose $A_1, \ldots, A_{n_2^2}$ such that $h_{i,j}$ is the identity matrix.} 
\STATE {Step 2: Set the parameter $\theta^i= -H(W_i)+H(W_{n_2^2})$ for $i=1, \ldots, n_2^2-1$, 
which is the solution of \eqref{MXP2}.}
%\STATE {Step 2: Find the parameter $\theta$ as the solution of \eqref{MXP2}.}
\STATE {Step 3: Calculate $\phi(\theta)$ by using \eqref{XZP2}.}
\STATE {Step 4: Calculate $\widehat{Q}_{X,*}(x):= \Tr \rho_\theta A_j$,
where $\rho_{\theta}$ is calculated by \eqref{nat2}.}
\STATE {Step 5: If the condition \eqref{MMR2} holds,
%$\widehat{Q}_{X,*}(x)\ge 0 $ for $x \in \X$, 
we consider that \eqref{CXP} holds and
output $\phi(\theta)-H(W_n)$ as the capacity.
Otherwise, we output ``the capacity cannot be computed.''}
\end{algorithmic}
\end{algorithm}

\section{Comparison}\Label{SC}
In the calculation of the capacity of classical channel,
when an error $\epsilon$ is allowed,
the conventional method \cite{Arimoto,Blahut} has calculation amount 
$O(\frac{n_1 n_2  \log n_1 }{\epsilon})$ because 
each iteration has calculation amount $n_1 n_2 $ and the number of iterations is 
$O(\frac{ \log n_1 }{\epsilon})$.
%Also, the recent paper \cite{SSML} improves it to $O(\frac{n_1^2 n_2  \sqrt{\log n_1 }}{\epsilon})$.
While it is smaller  than our method (Algorithm \ref{protocol1}) when $n_1=n_2$,
our method derives the exact value of the maximum without iteration.

When only Condition (A) holds, we can consider to solve the minimization \eqref{MIH} due to Theorem \ref{TL}.
However, it is difficult to analytically solve \eqref{MIH} in general.
Since this method needs larger calculation amount to obtain 
$\theta^1,\ldots, \theta^{n_1-1}$,
the algorithm based on Theorem \ref{TL} does not have advantage over 
the conventional method \cite{Arimoto,Blahut} except for the case
that the minimization \eqref{MIH} is analytically solved.

Next, we compare Algorithm \ref{protocol5} with existing algorithms for the capacity of a classical-quantum channel.
The algorithm by \cite{Nagaoka,RISB}
has calculation complexity $ O(\frac{ (n_1 n_2^2+n_2^3)\log n_1}{\epsilon}+ n_1 n_2^3)$
The algorithm by \cite{Sutter}
has calculation complexity $ O(\frac{ \max(n_1 ,n_2) n_2^3 \sqrt{\log n_1}}{\epsilon})$.
Unfortunately, these existing algorithms are smaller than our method, Algorithm\ref{protocol5} when $n_1=n_3(n_2-1)+1$ or $n_1=n_2^2$.
However, our method derives the exact value of the maximum without iteration
when we calculate the inverse matrix exactly.
This point is an advantage over existing methods.

Indeed, in practice, to evaluate the precision of our algorithm,
we need to evaluate the precision for each step including the calculations of the inverse matrix, logarithm, and exponential.
Such an analysis is left for a future study.

\if0
Although our algorithm requires Condition (D),
our method can be applied to the case with the modified condition (D') by relaxing 
the first condition in (D) in the same way as the above case.
\fi

%\cite{Dupuis,Sutter,Li}
\section{Conclusion and future study}\Label{SF}
We have proposed an exact algorithm to calculate the channel capacities of classical and classical-quantum channels.
However, we have various conditions to apply our algorithm.
Therefore, it is a future problem to remove conditions.
Indeed, Toyota \cite{Shoji} studied information geometrical structure \cite{Amari-Nagaoka}
for Arimoto-Blahut algorithm for 
the capacity of a classical channel.
Hence, it is an interesting topic to derive an information theoretical characterization of our method.

Further, it is a challenging problem to extend our method to 
the maximization of Gallager's function, i.e., R\'{e}nyi mutual information, including classical-quantum setting, which is related to the exponential decreasing rate \cite{Gallager,FNY} of the decoding error probability 
and the strong converse exponent \cite{Arimoto2,ON,MO}.
As another future study, we can consider an extension of our algorithm to wire-tap channel capacity \cite{Wyner,CK79,Yasui,YSM}.

%\cite{Arimoto,Blahut,Matz,Yu,SSML,NWS}

\section{Additional discussion}
After completing the review process of this paper,
the author found the reference \cite{Muroga} that has already derived 
an analytical calculation method for the channel capacity  
under a certain condition.
The reference \cite{NN} derived the same method as \cite{Muroga}.
The method by \cite{Muroga,NN} is the following;
First, assume $n_1=n_2$ and
the existence of the inverse matrix 
$(g_{x,x'})_{x,x'}$
of the transition matrix $(W(y|x))_{x,y}$, i.e.,
$\sum_{x'}g_{x,x'} W(y|x')=\delta_{y,x}$.
Then, the capacity is calculated as
\begin{align}
C=\log \Big( \sum_{x'=1}^{n_1}G_{x'} \Big),
\end{align}
where
\begin{align}
G_{x'}:= \sum_{x,y} g_{s,x'}W(y|x)\log W(y|x).
\end{align}
In addition, the input distribution $P_*$ realizing the capacity 
is given as
\begin{align}
P_*(x)= \exp(-C) \sum_{x'=1}^{n_1}g_{x,x'}\exp(G_{x'}).
\end{align}

Our method has the following advantage over the above method.
First, our method works even with classical-quantum channel
while their method works only with classical channel.
Second, their method needs to assume the existence of 
 the inverse matrix of the transition matrix $(W(y|x))_{x,y}$.
Although Algorithm \ref{protocol1} requires the existence of 
 the inverse matrix of the transition matrix $(W(y|x))_{x,y}$ in Step 1,
 our method can relax this condition in the following way
 because it is sufficient to find functions $f_1,\ldots, f_{n_1-1}$
 satisfying the conditions \eqref{CO} and \eqref{XPA}
 with $h_{i,j}=\delta_{i,j}$ and $n_2=n_1$.
 
Now, instead of the existence of 
 the inverse matrix of the transition matrix $(W(y|x))_{x,y}$,
we assume the existence of the inverse matrix of the matrix
$(W(y|x)- \frac{W(n_1|x)W(y|n_1)}{W(n_1|n_1)}
)_{x,y=1, \ldots, n_1-1} $
by $c_{j,y}$, i.e., $ \sum_{y=1}^{n_1-1}c_{j,y}
(W(y|x)- \frac{W(n_1|x)W(y|n_1)}{W(n_1|n_1)}
)=\delta_{x,j}$.
Then, we set $f_{1}, \ldots f_{n_1-1}$ as
$f_{j}(y)=c_{j,y}$ for $y=1, \ldots, n_1-1$,
$f_{j}(n_1)=-\sum_{y=1}^{n_1-1}
c_{j,y}\frac{W(y|n_1)}{W(n_1|n_1)}$, 
and 
$f_{j}(y)=0$ for $y=n_1, \ldots, n_2$.
We find that the functions $f_1,\ldots, f_{n_1-1}$
 satisfy the conditions \eqref{CO} and \eqref{XPA}
 with $h_{i,j}=\delta_{i,j}$.
Since
the existence of the inverse matrix of the matrix
$(W(y|x)- \frac{W(n_1|x)W(y|n_1)}{W(n_1|n_1)}
)_{x,y=1, \ldots, n_1-1} $
is a weaker condition than the existence of 
 the inverse matrix of the transition matrix $(W(y|x))_{x,y}$,
our method is better than the method by \cite{Muroga} even for the classical channel.

\section*{Acknowledgments}
The author is very grateful to Mr. Shoji Toyota for helpful discussions.
%In particular, he explained me the importance of the problem setting \eqref{MOA4} 
%in the relation to wire-tap channel \cite{Toyota}.

\appendices
\section{Summary for information geometry}\Label{SUMA}
To show Theorems \ref{T0} and \ref{T-1}, we summarize basic knowledge for 
information geometry, which was established in the reference \cite{Amari-Nagaoka}.
The following contents are used in Appendices \ref{A2} and \ref{A3}.
Given a finite probability space $\X$, % and a distribution $P_X$ on $\X$, 
we define an exponential family as follows.
Consider $l$ linearly independent random variables $f_1, \ldots, f_l$ on $\X$.
We define the distribution $P_{\theta,X}$ as
\begin{align}
P_{\theta,X}(x):=  e^{\sum_{j=1}^l \theta^j f_j(x) -\phi(\theta)},
\end{align}
where $\phi(\theta):= \log \sum_{x \in \X} e^{\sum_{j=1}^l \theta^j f_j(x)}$.
The set ${\cal E}:=\{P_{\theta,X} | \theta \in \mathbb{R}^l\}\subset {\cal P}_{\X}$ is called
an exponential family generated by random variables $f_1, \ldots, f_l$.
Also, the set
\begin{align}
{\cal M}:= \{ Q_X\in {\cal P}_{\X}| Q_X \hbox{ satisfies } \eqref{XXP}. \}
\end{align}
is called the mixture family generated by the constraint
\begin{align}
\sum_{x \in \X} f_j(x) Q_X(x)=a_j \Label{XXP}.
\end{align}
The following is a typical example of a mixture family.
For a subset $\X_0\subset \X$, as a generalization of ${\cal M}_0$, 
we define the mixture family ${\cal M}_{\X_0}$ as
\begin{align}
&{\cal M}_{\X_0}\nonumber \\
:= &\Big\{ Q_Y \in \P_{\Y}\Big|
Q_Y= \sum_{x \in \X\setminus \X_0 } c(x)W_{x}, ~
\sum_{x \in \X\setminus \X_0 } c(x)=1\Big\}.
\end{align}
When $\X_0$ is the empty set, ${\cal M}_{\X_0}$ coincides with ${\cal M}_{0}$.
Also, we simplify ${\cal M}_{\{x\}}$ to ${\cal M}_x$.

The following is known as Pythagorean theorem \cite{Amari-Nagaoka}.
\begin{theorem}\Label{TH6}
There uniquely exists an element $P_{X,*} \in {\cal E} \cap {\cal M}$.
Any elements $P_{X,1} \in {\cal M}$ and $P_{X,2} \in {\cal E}$ satisfy
\begin{align}
D( P_{X,1} \| P_{X,2})=D( P_{X,1} \| P_{X,*})+D( P_{X,*} \| P_{X,2}).
\end{align}
\end{theorem}

Using this theorem, we can show the following corollaries.
\begin{corollary}\Label{Cor1}
Given a distribution $Q_X $ on $\X$, there uniquely exists an element 
$Q_{X,*} \in {\cal M} $ such that
\begin{align}
D( P_{X,1} \| Q_{X})=D( P_{X,1} \| Q_{X,*})+D( Q_{X,*} \| Q_{X}) \Label{XOT}
\end{align}
for any element $P_{X,1} \in {\cal M}$.
$Q_{X,*}$ is called the projection of $Q_X $ to ${\cal M}$, and is denoted by 
$\Gamma_{{\cal M}}^{(m)}(Q_X)$.
\end{corollary}

\begin{corollary}\Label{Cor2}
Given a distribution $Q_X $ on $\X$, there uniquely exists an element 
$Q_{X,*} \in {\cal E} $ such that
\begin{align}
D( Q_{X} \| P_{X,2})=D( Q_{X} \| Q_{X,*})+D( Q_{X,*} \| P_{X,2}) \Label{XOT}
\end{align}
for any element $P_{X,2} \in {\cal E}$.
$Q_{X,*}$ is called the projection of $Q_X $ to ${\cal E}$, and is denoted by $\Gamma_{{\cal E}}^{(e)}(Q_X)$.
\end{corollary}

Now, we consider a one-parameter exponential family $\{P_t\}$.
\begin{lemma}\Label{L2}
For $t_1\le t_2\le t_3$, we have
\begin{align}
D(P_{t_1}\| P_{t_2} )+ D(P_{t_2}\| P_{t_3} ) \le
D(P_{t_1}\| P_{t_3} ).\Label{XAR}
\end{align}
\end{lemma}
\begin{proof}
Let $J_t$ be the Fisher information in the one-parameter exponential family $\{P_t\}$.
Then, we have
\begin{align}
D(P_{t}\| P_{t'} )=\int_{t'}^s J_s (s-t')ds. \Label{XNY}
\end{align}
The expression \eqref{XNY} implies \eqref{XAR}.
\end{proof}

\section{Proof of Lemma \ref{LAP}}\Label{A1}
We show this lemma by contradiction. We assume that 
$ D( W_x\| Q_Y) $ depends on $x \in \supp(Q_X)$. Then, 
the set 
$\X_0:= \{x_0 \in \X|  D( W_{x_0}\| Q_Y) < \max_{x \in \X} D( W_x\| Q_Y)  \}$
is not empty.
With a small $\epsilon >0$, we choose $Q_{X,\epsilon}$ as
\begin{align}
Q_{X,\epsilon}(x_0)&:= Q_{X}(x_0)- \frac{\epsilon}{|\X_0|}\\
Q_{X,\epsilon}(x')&:= Q_{X}(x')+ \frac{\epsilon}{|\X \setminus \X_0|}
\end{align}
for $x_0 \in \X_0$ and $x' \in \X \setminus \X_0$.
The above choices of $Q_{X,\epsilon}$ guarantee that
$W\cdot Q_{X,\epsilon}$ is closer to $W_{x'}$ than $W\cdot Q_{X}$ for 
$x' \in \X \setminus \X_0$, which implies
the relation
\begin{align}
D(W_{x'}\| W\cdot Q_{X,\epsilon}) < D(W_{x'}\| W\cdot Q_{X})\Label{ACAI}.
\end{align}
Since $W\cdot Q_{X}$ is closer to $W_{x_0}$ than $W\cdot Q_{X,\epsilon}$
for $x_0 \in \X_0$, we have
\begin{align}
D(W_{x_0}\| W\cdot Q_{X}) 
< D(W_{x_0}\| W\cdot Q_{X,\epsilon}).
\end{align}
Since $D(W_{x_0}\| W\cdot Q_{X}) < D(W_{x'}\| W\cdot Q_{X})$, we can choose 
a sufficiently small $\epsilon>0$ such that 
\begin{align}
 D(W_{x_0}\| W\cdot Q_{X}) 
<& D(W_{x_0}\| W\cdot Q_{X,\epsilon})\nonumber \\
<&
D(W_{x'}\| W\cdot Q_{X}).\Label{CAW}
\end{align}
The relations \eqref{ACAI} and \eqref{CAW} imply
\begin{align}
&\max_{x \in\X}D(W_{x}\| W\cdot Q_{X,\epsilon}) \nonumber \\
<& D(W_{x'}\| W\cdot Q_{X})=\max_{x \in\X} D(W_{x}\| W\cdot Q_{X}).\Label{ASS}
\end{align}
However, $W \cdot Q_{X}$ is the minimizer of \eqref{MOA}, which contradicts 
\eqref{ASS}.

\section{Proof of Lemma \ref{LL2}}\Label{ALL2}
%\begin{proof}
We choose $n_2-1$ linearly independent 
functions $f_1, \ldots, f_{n_2-1}$ on $\Y$ such that
they are not constant function and 
\begin{align} 
\sum_{y \in \Y}W_x(y) f_j(y)=0,\quad 
\sum_{y \in \Y}W_{n_1}(y) f_{j'}(y)=0
\end{align} 
for $j=n_1, \ldots, n_2-1$, $j'=1, \ldots, n_2-1$, and $x=1, \ldots, n_1$.
In fact, the set ${\cal M}_0$ is rewritten as 
\begin{align}
&{\cal M}_0\nonumber \\
=& \Big\{Q_Y\!\in\! \P_{\Y} \Big|
\sum_{y \in \Y}Q_Y(y) f_j(y)=0 \hbox{ for }j=n_1, \ldots, n_2-1
\Big\}.
\end{align} 
Then, %by using a distribution $P_Y$ on $\Y$, 
any distribution on ${\cal Y}$ is parameterized as
\begin{align}
P_{\theta,Y}(y):= e^{\sum_{j=1}^{n_2-1} \theta^j f_j(y) -\phi(\theta)},
\end{align}
where $\phi(\theta):= \log \sum_{y \in \Y} e^{\sum_{j=1}^{n_2-1} \theta^j f_j(y)}$.
For any vector $\theta_1=(\theta^1, \ldots, \theta^{n_1-1})$,
there exist parameters $\theta_2(\theta_1)=
(\theta^{n_1}(\theta_1), \ldots, \theta^{n_2-1}(\theta_1))$
such that $P_{(\theta_1,\theta_2(\theta_1)),Y} \in {\cal M}_0$.
This fact can be shown as follows.
Given a vector $\theta_1$, we define the set
\begin{align}
{\cal G}(\theta_1):=\Big\{\Big(\sum_{y \in \Y} 
P_{(\theta_1,\theta_2),Y}(y) f_j(y) \Big)_{j=n_1}^{n_2-1}
\Big| \theta_2 \in \mathbb{R}^{n_2-n_1}\Big\}.
\end{align}
This set equals the inner of 
the convex hull of $\{ ( f_j(y) )_{j=n_1}^{n_2-1} \}_{y \in \Y}$.
That is, the set ${\cal G}(\theta_1)$ does not depend on $\theta_1 \in \mathbb{R}^{n_1-1}$.
The first equation shows that the origin $(0,\ldots, 0)$ belongs to 
$\cup_{\theta_1 \in \mathbb{R}^{n_1-1}} {\cal G}(\theta_1)$.
Hence, 
the origin $(0,\ldots, 0)$ belongs to 
${\cal G}(\theta_1)$ for an element $\theta_1 \in \mathbb{R}^{n_1-1}$.
Therefore, there exist parameters $\theta_2(\theta_1)=
(\theta^{n_1}(\theta_1), \ldots, \theta^{n_2-1}(\theta_1))$
such that $P_{(\theta_1,\theta_2(\theta_1)),Y} \in {\cal M}_0$.

Then, we choose the parameters $h_{x,j}$ as
\begin{align}
h_{x,j}:= \sum_{y \in \Y}W_x(y) f_j(y)
\end{align}
for $j=1, \ldots, n_1-1$.
Since functions $f_1, \ldots, f_{n_1-1}$ are linearly independent, % and $n_1=n_2$,
due to Condition (A),
%there exists an element $x_0 \in \X$ such that
the vectors $\{(h_{x,j})_{j=1}^{n_1-1} \}_{x=1}^{n_1-1}$ 
are linearly independent.

Then, we have
\begin{align}
D( W_x\| P_{\theta,Y})= -H(W_x) -\sum_{j=1}^{n_1-1} \theta^j h_{x,j} 
-\phi(\theta_1,\theta_2(\theta_1)).
\end{align}
Condition (B) with $Q_Y=P_{\theta,Y}$ is rewritten as
\begin{align}
&-H(W_x) -\sum_{j=1}^{n_1-1} \theta^j h_{x,j} -\phi(\theta_1,\theta_2(\theta_1))\nonumber \\
%=-H(W_{n_1}) -\sum_{j=1}^{n_1-1} \theta^j h_{n_1,j} -\phi(\theta_1,\theta_2(\theta_1))
=&-H(W_{n_1}) -\phi(\theta_1,\theta_2(\theta_1))
\end{align}
for $x=1, \ldots, n_1$. This condition is rewritten as
\begin{align}
-H(W_x)+H(W_{n_1})
= \sum_{j=1}^{n_1-1}h_{x,j} \theta^j  \Label{SASC}
\end{align}
for $x=1, \ldots, n_1-1$.
Since the matrix $h_{x,j}$ is invertible,
only one vector $\theta_1=(\theta^1, \ldots, \theta^{n_1-1})  $ 
satisfies \eqref{SASC}, i.e., Condition (B).
%\end{proof}

\section{Proof of Theorem \ref{T0}}\Label{A2}
Assume the condition (ii). 
%Since $Q_X \mapsto D(W_x\| W\cdot Q_X)$ is convex, the map
%$Q_X \mapsto \max_{x \in \X} D(W_x\| W\cdot Q_X)$ is also convex.
%Hence, the local minimum of  the above function achieves the global minimum.
For $Q_X(\neq \widehat{Q}_{X,*}) \in \P_{\X}$, we have 
$ \max_{x \in \X} D(W_x\| W\cdot Q_X) > \max_{x \in \X} D(W_x\| W\cdot \widehat{Q}_{X,*})$ 
because $W\cdot Q_X$ and belongs to ${\cal M}_0$ and 
only one element of ${\cal M}_0$ satisfy Condition (B) due to Lemma \ref{LL2}.
Hence, 
$ Q_{Y,*}$ achieves $C(W)$, which implies Condition (i).

Assume the condition (i). 
There exists $Q_X \in P_{\X}$ such that 
$D( W_x\| Q_{Y,*})=\sum_{x\in \X}Q_X(x) D(W_x \| W \cdot Q_X  )$.
For any element $x \in \supp(Q_X)$, we have 
\begin{align}
 D(W_x \| W \cdot Q_X  )=D( W_x\| Q_{Y,*}).\Label{COX}
\end{align}
Then, the distribution $\overline{Q}_{Y,*}:= \Gamma_{{\cal M}_{0}}^{(m)}(Q_{Y,*})$ satisfies 
$ D( W_x\| Q_{Y,*})= D( W_x\| \overline{Q}_{Y,*})+D( \overline{Q}_{Y,*}\| Q_{Y,*})$, 
where
the projection $\Gamma_{{\cal M}_{0}}^{(m)}$ 
is defined in Appendix \ref{SUMA}.
Hence, $ D( W_x\| Q_{Y,*})\ge D( W_x\| \overline{Q}_{Y,*})$.

Since $\min_{Q_Y \in {\cal P}_{\Y} }\max_{x \in \X} D( W_x\| Q_Y)=
\max_{x \in \X} D( W_x\| W \cdot Q_X)$, we have
$ D( W_x\| W \cdot Q_X)
= D( W_x\| Q_{Y,*})= D( W_x\| \overline{Q}_{Y,*})$ for $x \in \supp(Q_X)$.
Hence, $D( \overline{Q}_{Y,*}\| Q_{Y,*})=0$, i.e., $\overline{Q}_{Y,*}= Q_{Y,*}$.
That is, $Q_{Y,*}$ belongs to ${\cal M}_{0}$.
Due to Condition (A) and Lemma \ref{LL2}, 
the condition \eqref{COX} uniquely determines $Q_{Y,*}$.
Hence,  $W \cdot\widehat{Q}_{X,*}=Q_{Y,*}=W \cdot Q_X$.
Condition (A) guarantees the relation $\widehat{Q}_{X,*}=Q_X$, which implies the condition (ii).

\section{Proof of Theorem \ref{T-1}}\Label{A3}
%We show \eqref{ZP1} by the contradiction.
Due to Condition (A), there uniquely exists a distribution $Q_{X,*}
\in \P_{\X}$ to achieve the capacity $C(W)$.
It is sufficient to show that 
$Q_{X,*}(x_0)= 0$ for any element $x_0 \in {\cal N}(\widehat{Q}_{X,*})$.
For this aim, we fix an arbitrary element $x_0 \in {\cal N}(\widehat{Q}_{X,*})$.

\noindent{\bf Step 1:}\quad
We show that there exists a distribution $Q_{Y,0} \in {\cal M}_{x_0}$
such that 
\begin{align}
\max_{x \in \X} D( W_x\| Q_{Y,*}) 
\ge \max_{x \in \X} D( W_x\| Q_{Y,0})
= D( W_{x'}\| Q_{Y,0}) \Label{COC}
\end{align}
for any element $x' \in \X\setminus\{x_0\}$.

Since $Q_{Y,*} $ is the unique element of ${\cal M}_0$ to satisfy Condition B, any element $x' \in \X\setminus\{x_0\}$ satisfies
\begin{align}
D( W_{x_0}\| Q_{Y,*}) =D( W_{x'}\| Q_{Y,*}) \Label{XXM}.
\end{align}
We choose a function $f_{x_0}$ on $\X$ such that
\begin{align}
\sum_{y \in \Y}f_{x_0}(y)W_{x_0}(y) &=1 , \Label{CP1} \\
\sum_{y \in \Y}f_{x_0}(y)W_{x}(y) &=0\Label{CP2}
\end{align}
for any element $x (\neq x_0) \in \X$.
We denote $- \widehat{Q}_{X,*}(x_0)>0$ by $a$.
Then, we have 
\begin{align}
\frac{1}{1+a}Q_{Y,*}+ \frac{a}{1+a} W_{x_0} \in {\cal M}_{x_0}.\Label{CP3}
\end{align}
The combination of \eqref{CP2} and \eqref{CP3} implies that
\begin{align}
\sum_{y \in \Y}f_{x_0}(y) 
(\frac{1}{1+a}Q_{Y,*}(y)+ \frac{a}{1+a} W_{x_0}(y))=0.
\Label{CP4}
\end{align}
Then, the combination of \eqref{CP1} and \eqref{CP4} yields that
\begin{align}
\sum_{y \in \Y}f_{x_0}(y) Q_{Y,*}(y)= -a.\Label{XM1}
\end{align}

The distribution $Q_{Y,0}:=\Gamma_{{\cal M}_{x_0}}^{(m)} (Q_{Y,*}) \in {\cal M}_{x_0} $ 
satisfies
\begin{align}
D(Q_Y \| Q_{Y,*} )= D(Q_Y \| Q_{Y,0} )+D(Q_{Y,0} \| Q_{Y,*} )
\Label{CAS2}
\end{align}
for any distribution $Q_Y \in {\cal M}_{x_0} $.
We define the exponential family 
${\cal E}_1:=\{Q_{Y,t}\}_{t \in \mathbb{R}}$ as
\begin{align}
Q_{Y,t}(y):= Q_{Y,0}(y) e^{t f_{x_0}(y)-\varphi(y)},
\end{align}
where
\begin{align}
\varphi(y):=\log \sum_{y \in \Y} Q_{Y,0}(y) 
e^{t f_{x_0}(y)}.
\end{align}
Hence, $Q_{Y,0}$ coincides with the case with $t=0$.
We choose $t_*$ such that $Q_{Y,t_*}= Q_{Y,*}$.
The relation \eqref{XM1} guarantees that $t_*<0$.
Also, we choose $t_0$ as 
$Q_{Y,t_0}=\Gamma_{{\cal E}_1}^{(e)}(W_{x_0}) $. Then, we have
\begin{align}
D(W_{x_0} \| Q_{Y,t} )= D(W_{x_0} \| Q_{Y,t_0} )+D(Q_{Y,t_0} \| Q_{Y,t} ).
\Label{XLT2}
\end{align}
for any $t_0 \in \mathbb{R}$.
The relation \eqref{XM1} guarantees that $t_0>0$.
Since $t_*<0$ and $t_0>0$, Lemma \ref{L2} yields that
\begin{align}
D( Q_{Y,t_0} \| Q_{Y,0}) \le D( Q_{Y,t_0} \| Q_{Y,t_*})-D( Q_{Y,0} \| Q_{Y,t_*})
. \Label{XLT}
\end{align}

\begin{figure}[htbp]
%\centering
%\includegraphics[scale=0.4]{MHRepeater.PNG}
\begin{center}
 \includegraphics[width=0.95\linewidth]{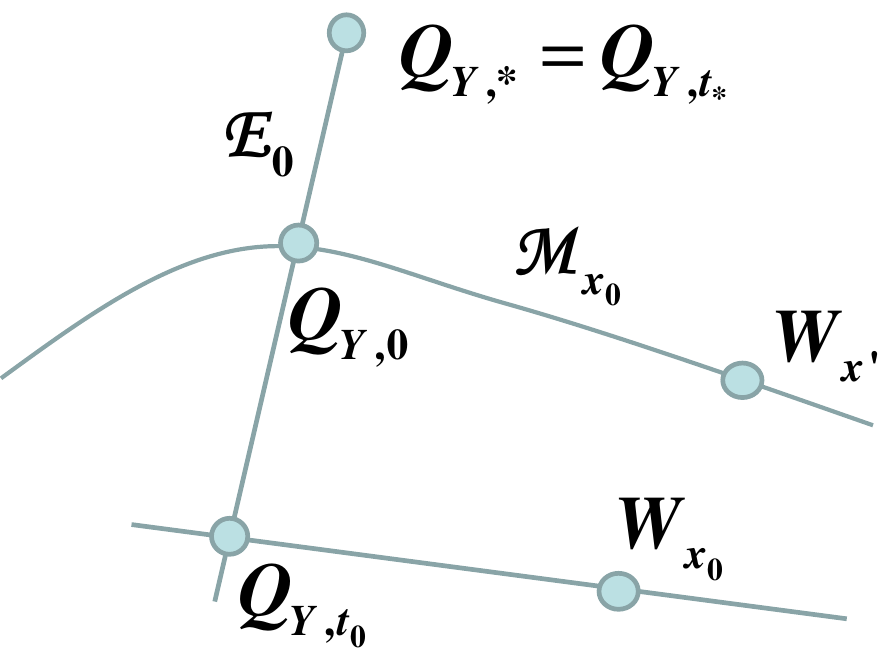}
  \end{center}
\caption{Relation among various distributions appearing in Step 1 of the proof of Theorem \ref{T-1}.
This figure shows the topological relation among the distributions
 $Q_{Y,*}=Q_{Y,t_*}, Q_{Y,0}$, $Q_{Y,t_0}$, 
the exponential family ${\cal E}_0$, and 
the mixture family ${\cal M}_{x_0}$.}
\Label{relation}
\end{figure}

The combination of \eqref{XLT2} and \eqref{XLT} guarantees that
\begin{align}
& D(W_{x_0} \| Q_{Y,0})
\stackrel{(a)}{=} D(W_{x_0} \| Q_{Y,t_0} )+D( Q_{Y,t_0} \| Q_{Y,0}) \nonumber \\
\stackrel{(b)}{\le} &D(W_{x_0} \| Q_{Y,t_0} )+D(Q_{Y,t_0} \| Q_{Y,t_*} )-D( Q_{Y,0} \| Q_{Y,t_*}) \nonumber \\
\stackrel{(c)}{=}&D( W_{x_0} \| Q_{Y,t_*})-D( Q_{Y,0} \| Q_{Y,t_*}) \nonumber \\
\stackrel{(d)}{=}&D( W_{x_0} \| Q_{Y,*})-D( Q_{Y,0} \| Q_{Y,*}) \nonumber \\
\stackrel{(e)}{=}&D( W_{x'} \| Q_{Y,*})-D( Q_{Y,0} \| Q_{Y,*})\stackrel{(f)}{=}D( W_{x'} \| Q_{Y,0})
%D( W_{x_0} \| Q_{Y,0}) < D( W_{x_0} \| Q_{Y,t_*})
\Label{CAS1}
\end{align}
for any element $x' \in \X\setminus\{x_0\}$.
Each step is shown in the following way.
Steps $(a)$ and $(c)$ follow from \eqref{XLT2}.
Step $(b)$ follows from \eqref{XLT}.
Step $(d)$ follows from $Q_{Y,t_*}=Q_{Y,*} $.
Step $(e)$ follows from \eqref{XXM}.
%$D( W_{x_0} \| Q_{Y,*})=D( W_{x'} \| Q_{Y,*}) $.
Step $(f)$ follows from \eqref{CAS2}.
\eqref{CAS1} shows the following two facts.
One is $D( W_{x'} \| Q_{Y,0})$ does not depend on $x' \in \X\setminus\{x_0\}$.
The other is $D(W_{x_0} \| Q_{Y,0})\le D( W_{x'} \| Q_{Y,0})$.
The combination of these two facts implies
\begin{align}
&\max_{x \in \X} D( W_x\| Q_{Y,0})
=D(W_{x'} \| Q_{Y,0})\nonumber \\
\stackrel{(a)}{\le} &D( W_{x'} \| Q_{Y,*})
\le \max_{x \in \X} D( W_x\| Q_{Y,*}) ,
\end{align}
where Step $(a)$ follows from \eqref{CAS2}.
Hence, we obtain \eqref{COC}.

\noindent{\bf Step 2:}\quad
We choose a function $\widehat{Q}_{X,1}$ on $\X\setminus \{x_0\}$
such that 
\begin{align}
\sum_{x \in\X\setminus \{x_0\}} \widehat{Q}_{X,1}(x) W_x= Q_{Y,1},
\Label{XAP}
\end{align}
where $\widehat{Q}_{X,1}$ uniquely exists because $ Q_{Y,1}\in {\cal M}_{x_0}$.
We show the desired statement $Q_{X,*}(x_0)=0$
when $\widehat{Q}_{X,1}(x)\le 0$ for $x \in \X\setminus \{x_0\}$.

In this case, it is sufficient to show that 
$Q_{X,*}(x_0)=\widehat{Q}_{X,1}$, i.e., 
$\widehat{Q}_{X,1}$ achieves the capacity $C(W)$.
We have
\begin{align}
&\sum_{x \in \X\setminus \{x_0\}} Q_{X,1}(x) D( W_x\| Q_{Y,1})
\stackrel{(a)}{=} \max_{x \in \X\setminus \{x_0\}} D( W_x\| Q_{Y,1}) \nonumber \\
\stackrel{(b)}{=}&\max_{Q_X \in {{\cal P}_{\X \setminus \{x_0\} }} }
\sum_{x\in \X\setminus \{x_0\}}Q_X(x) D(W_x \| W \cdot Q_X  ) \nonumber \\
\stackrel{(c)}{=}&\min_{Q_X \in {{\cal P}_{\X \setminus \{x_0\} }} }
\max_{x\in \X\setminus \{x_0\}} D(W_x \| W \cdot Q_X  ) \nonumber \\
\stackrel{(d)}{=}&\min_{Q_Y \in {\cal M}_{x_0}} 
\max_{x\in \X\setminus \{x_0\}} D(W_x \| Q_Y  ) \nonumber \\
\stackrel{(e)}{\le} &\min_{Q_X \in {{\cal P}_{\X  }} }
\max_{x\in \X\setminus \{x_0\}} D(W_x \| W \cdot Q_X  )  \nonumber\\
{\le} &\min_{Q_X \in {{\cal P}_{\X  }} }
\max_{x\in \X} D(W_x \| W \cdot Q_X  )  \nonumber\\
{\le} &
\max_{x\in \X} D(W_x \| W \cdot Q_{X,1}  ) .
\end{align}
In the above relations, ${\cal P}_{\X \setminus \{x_0\}}$ means
the set of probability distributions on the set $\X \setminus \{x_0\} $.
Each step is shown in the following way. % while Steps $(f)$ and $(g)$ are trivial.
Step $(a)$ follows from the second equation in \eqref{COC}.
Step $(b)$ follows from Theorem \ref{T0}.
Step $(c)$ follows from \eqref{MOA}.
Step $(d)$ is shown as follows.
Since $Q_Y \mapsto D(W_x \| Q_Y  ) $ is convex, 
$Q_Y\mapsto \max_{x\in \X\setminus \{x_0\}} D(W_x \| Q_Y  ) $ is
also convex.
Since $Q_{Y,1} $ achieves a local minimum, it also achieve the global 
minimum in ${\cal M}_{x_0}$.

Step $(e)$ is shown as follows.
For $Q_X \in \P_{\X}$, 
the distribution $Q_{Y}':= \Gamma_{{\cal M}_{x_0}}^{(m)}(W\cdot Q_X)$ satisfies
\begin{align}
&D(W_x \| W \cdot Q_X  ) =
D(W_x \| Q_{Y}' ) +D(Q_{Y}' \| W \cdot Q_X  ) \nonumber \\
\ge & D(W_x \| Q_{Y}' ) 
\hbox{ for } x \in \X\setminus \{x_0\},
\end{align}
which shows $(e)$.

\if0
\begin{align}
\min_{Q_Y \in {\cal M}_{x_0}} 
\max_{x\in \X\setminus \{x_0\}} D(W_x \| Q_Y  ) 
\le 
\min_{Q_X \in {{\cal P}_{\X  }} }
\max_{x\in \X\setminus \{x_0\}} D(W_x \| W \cdot Q_X  )  ,\Label{XZTP}
\end{align}
\fi

Hence, we have 
\begin{align}
C(W)=\sum_{x \in \X\setminus \{x_0\}} Q_{X,1}(x) D( W_x\| W \cdot Q_{X,1}).
\end{align}

\noindent{\bf Step 3:}\quad
We show the desired statement $Q_{X,*}(x_0)=0$
when there exists $x_1 \in \X\setminus \{x_0\}$ such that
$Q_{X,1}(x_1)< 0$.
Applying the same discussion as Step 1
with replacing $Q_{Y,*} $ and $Q_{Y,0}$
by $Q_{Y,1} $ and $Q_{Y,2}$, respectively,
 we find that there exists a distribution 
$Q_{Y,2} \in {\cal M}_{\{x_0,x_1\}}$
such that 
\begin{align}
\max_{x \in \X} D( W_x\| Q_{Y,1}) 
\ge \max_{x \in \X} D( W_x\| Q_{Y,2})
= D( W_{x'}\| Q_{Y,2}) \Label{COC2}
\end{align}
for any element $x' \in \X\setminus\{x_0,x_1\}$.
Then, we choose $\widehat{Q}_{X,2}$ in the same way as \eqref{XAP}.
If $\widehat{Q}_{X,2}(x)\ge 0$ for $x \in \X\setminus \{x_0,x_1\}$,
we find that $Q_{X,*}(x_0)=0$ in the same way as Step 2.
Otherwise, 
we repeat the above procedure up to $i$ times until
we have $\widehat{Q}_{X,i}(x)\ge 0$ for $x \in \X\setminus \{x_0,x_1, \ldots, x_{i-1}\}$.
Once we obtain the above condition,
we find $Q_{X,*}(x_0)=0$ in the same way as Step 2.

\section{Proof of Lemma \ref{LX1}}\Label{A4}
To show Lemma \ref{LX1}, we prepare functions $\overline{f}_1,\ldots, \overline{f}_{n_2-1}$ to satisfy 
the condition in Theorem \ref{TL}.
We denote the distribution defined in \eqref{nat} based on these functions $\overline{f}_1,\ldots, \overline{f}_{n_2-1}$
by $\overline{P}_{\theta,Y}$.
Such functions are given as linear combination of the original functions ${f}_1,\ldots, {f}_{n_2-1}$
by using coefficient $a^j_{j'}$ as
\begin{align}
\sum_j {f}_j a^j_{j'}= \overline{f}_{j'}.
\end{align}
Hence, we have
\begin{align}
\sum_{j'=1}^{n_2-1} \overline{f}_{j'}(y) \theta^{j'}=
\sum_{j=1}^{n_2-1} {f}_{j}(y)
\bigg(\sum_{j'=1}^{n_2-1} a^j_{j'} \theta^{j'}\bigg). 
\end{align}
Using this relation, we find that $\overline{P}_{\theta,Y}={P}_{\overline{\theta},Y}$, where
$\overline{\theta}^j=\sum_{j'=1}^{n_2-1} a^j_{j'} \theta^{j'}$.
Thus, the set ${\cal E}_0$ can be characterized with the new functions $\overline{f}_1,\ldots, \overline{f}_{n_2-1}$.
Therefore, without loss of generality, we can assume that 
the functions ${f}_1,\ldots, {f}_{n_2-1}$ satisfies the condition in Theorem \ref{TL}.

We choose $\theta^1,\ldots, \theta^{n_1-1}$ satisfies the condition \eqref{XPA}.
Then, we have 
\begin{align}
&{\cal E}_0\nonumber \\
=&\{P_{ \theta^1,\ldots, \theta^{n_1-1}, \eta^{n_1},\ldots, \eta^{n_2-1},Y}|
(\eta^{n_1},\ldots, \eta^{n_2-1}) \in \mathbb{R}^{n_2-n_1}\}.
\end{align}
Hence, ${\cal E}_0$ is an exponential family generated by ${f}_{n_1},\ldots, {f}_{n_2-1}$.

Since $f_{i,j}=0$ for $j=n_1+, \ldots, n_2-1$,
${\cal M}_0$ can be written as
\begin{align}
&{\cal M}_0\nonumber \\
=&\{ Q_Y \in {\cal P}_{\Y}| \sum_{y \in \Y} f_j(y) Q_Y(y) =0 \hbox{ for } j=n_1, \ldots, n_2-1 \}. \Label{XPY}
\end{align}
Hence, 
${\cal M}_0$ is a mixture family generated by the same functions ${f}_{n_1},\ldots, {f}_{n_2-1}$.
Therefore, Theorem \ref{TH6} implies Lemma \ref{LX1}.

\begin{IEEEbiographynophoto}{Masahito Hayashi}(Fellow, IEEE) was born in Japan in 1971.
He received the B.S.\ degree from the Faculty of Sciences in Kyoto
University, Japan, in 1994 and the M.S.\ and Ph.D.\ degrees in Mathematics from
Kyoto University, Japan, in 1996 and 1999, respectively. He worked in Kyoto University as a Research Fellow of the Japan Society of the
Promotion of Science (JSPS) from 1998 to 2000,
and worked in the Laboratory for Mathematical Neuroscience,
Brain Science Institute, RIKEN from 2000 to 2003,
and worked in ERATO Quantum Computation and Information Project,
Japan Science and Technology Agency (JST) as the Research Head from 2000 to 2006.
He also worked in the Superrobust Computation Project Information Science and Technology Strategic Core (21st Century COE by MEXT) Graduate School of Information Science and Technology, The University of Tokyo as Adjunct Associate Professor from 2004 to 2007.
He worked in the Graduate School of Information Sciences, Tohoku University as Associate Professor from 2007 to 2012.
In 2012, he joined the Graduate School of Mathematics, Nagoya University as Professor.
In 2020, he joined Shenzhen Institute for Quantum Science and Engineering, Southern University of Science and Technology, Shenzhen, China
as Chief Research Scientist.

In 2011, he received Information Theory Society Paper Award (2011) for ``Information-Spectrum Approach to Second-Order Coding Rate in Channel Coding''.
In 2016, he received the Japan Academy Medal from the Japan Academy
and the JSPS Prize from Japan Society for the Promotion of Science.
he is IMS (Institute of Mathematical Statistics) Fellow and 
AAIA (the Asia-Pacific Artificial Intelligence Association) Fellow.

In 2006, he published the book ``Quantum Information: An Introduction''  from Springer, whose revised version was published as ``Quantum Information Theory: Mathematical Foundation'' from Graduate Texts in Physics, Springer in 2016.
In 2016, he published other two books ``Group Representation for Quantum Theory'' and ``A Group Theoretic Approach to Quantum Information'' from Springer.
He is on the Editorial Board of {\it International Journal of Quantum Information}.
His research interests include classical and quantum information theory and classical and quantum statistical inference.
\end{IEEEbiographynophoto}
\end{document}